\RequirePackage{fix-cm}
\documentclass[smallextended]{svjour3}       
\smartqed  
\usepackage{graphicx}
\usepackage{csquotes}
\usepackage[table]{xcolor}
\usepackage[authoryear]{natbib}
\usepackage{amsmath}
\usepackage{amssymb}
\usepackage{enumitem}
\usepackage{bbm}
\begin{document}

\title{Perfect Prediction in Normal Form: Superrational Thinking Extended to Non-Symmetric Games}
\titlerunning{Perfect Prediction in Normal Form}


\author{Ghislain Fourny}


\institute{G. Fourny \at
              ETH Z\"urich \\
              Department of Computer Science \\
              \email{ghislain.fourny@inf.ethz.ch}\\
}

\date{Originally December 15, 2017. Updated February 19, 2020.}

\maketitle

\begin{abstract}
This paper introduces a new solution concept for non-cooperative games in normal form with no ties and pure strategies: the Perfectly Transparent Equilibrium. The players are rational in all possible worlds and know each other's strategies in all possible worlds -- which, together, we refer to as Perfect Prediction. The anticipation of a player's decision by their opponents is counterfactually dependent on the decision, unlike in Nash Equilibria where the decisions are made independently. The equilibrium, when it exists, is unique and is Pareto optimal.

This equilibrium is the normal-form counterpart of the Perfect Prediction Equilibrium; the prediction happens ``in another room'' rather than in the past. The equilibrium can also be seen as a natural extension of Hofstadter's superrationality to non-symmetric games. Algorithmically, an iterated elimination of non-individually-rational strategy profiles is performed until at most one remains. An equilibrium is a strategy profile that is immune against knowledge of strategies in all possible worlds and rationality in all possible worlds, a stronger concept than common knowledge of rationality but also stronger than common counterfactual belief of rationality.

We formalize and contrast the Non-Nashian Decision Theory paradigm, common to this and several other papers, with Causal Decision Theory and Evidential Decision Theory. We define the Perfectly Transparent Equilibrium algorithmically and prove (when it exists) that it is unique, that it is Pareto-optimal, and that it coincides with Hofstadter's Superrationality on symmetric games. We relate the finding to concepts found in the literature such as Individual Rationality, Rationalizability, Minimax-Rationalizability, Second-Order Nash Equilibria, the Program Equilibrium, the Perfect Prediction Equilibrium, Shiffrin's Joint-Selfish-Rational Equilibrium, the Stalnaker-Bonanno Equilibrium, the Perfect Cooperation Equilibrium, the Translucent Equilibrium, the Correlated Equilibrium, and Quantum Games. Finally, we specifically discuss inclusion relationships on the special case of symmetric games between Individual Rationality, Minimax-Rationalizability, Superrationality, and the Perfectly Transparent Equilibrium, and contrast them with asymmetric games.

\end{abstract}

\section{Introduction}
\label{section-introduction}
\subsection{Superrational thinking}
\label{section-superrational-thinking}

On the planet Betazed, a member of the United Federation of Planets, people have telepathic powers \citep{Roddenberry1969}. The prisoner's dilemma is not one to Betazoids: they cooperate, as it is quite obvious to them. They can read each other's minds, which leads to strong dependencies between their decisions. The assumption underlying the Nash equilibrium \citep{JNNCG}, namely that the opponent's strategy is kept frozen and fixed while optimizing one's payoff, does not apply to them: indeed, any change of strategy leads to an instantaneous change of strategy of the opponent. Defect, and the opponent defects. Cooperate, and they cooperate as well.

While this scenario pertains to twenty-fourth-century science fiction, the last decade's progress in data science -- especially large scale data analysis and machine learning -- hints that the decisions of human beings, especially at large scales, can be predicted to some extent, and that extent increases every year. This observation shakes a fundamental axiom of neoclassical economics -- namely, that agents make decisions independently of each other. This axiom, which some refer to as free choice\footnote{We consider this mainstream approach to free choice a strong definition of free choice. There are other, weaker definitions of free choice that are compatible with being predictable, for example, that one ``could have acted otherwise''. This latter approach is taken in this paper.}, is also at the core of Nashian\footnote{Throughout this paper, we use the adjective ``Nashian'' to refer to the Nash paradigm. Usage of the term Non-Nashian should be seen as a tribute to the work by John Nash, similar to how we call Non-Euclidian Geometry as a reference to Euclid.} game theory and in particular of the Nash equilibrium.

\citet{Hofstadter1983} suggested an alternative line of reasoning, superrationality, which we introduce formally in Section \ref{section-superrationality}. His main idea is that the players have such a high level of awareness of their rationality -- and of their common knowledge thereof, and of their own reasonings -- that they are able to reason on a meta-level, taking into account that the opponent is reasoning in the exact same way. In Douglas Hofstadter's words:

\begin{displayquote}
``If reasoning dictates an answer, then everyone should independently come to that answer. Seeing this fact is itself the critical step in the reasoning toward the correct answer [...]. Once you realize this fact, then it dawns on you that \emph{either} all rational players will choose D \emph{or} all rational players will choose C. This is the crux.
 Any number of ideal rational thinkers faced with the same situation and undergoing similar throes of reasoning agony will necessarily come up with the identical answer eventually, so long as reasoning alone is the ultimate justification for their conclusion. Otherwise, reasoning would be subjective, not objective as arithmetics is. A conclusion reached by reasoning would be a matter of preference, not of necessity.''
\end{displayquote}

Douglas Hofstadter gave a concrete algorithm for superrational play that is specific to symmetric games in normal form. Superrational thinkers start with the assumption that there is only one rational strategy, and then use their logical skills to find it. Because the game is symmetric, the rational strategies must be identical on both sides, so that the final outcome will necessarily be on the diagonal. Knowing this, it is straightforward that the reasonable course of action is, for both players, to pick the strategy that leads to the optimal strategy profile on the diagonal (Figure \ref{fig-prisoner-dilemma-hofstadter}). This is the one rational strategy, confirming the initial assumption.

\begin{figure}
\begin{center}
\begin{tabular}{|r|c|c|}
\hline
& Defect & Cooperate\\
\hline
Defect & 1, 1 & 3, 0\\
\hline
Cooperate &  0, 3 & \cellcolor{black!25}2, 2\\
\hline
\end{tabular}
\end{center}
\caption{The prisoner's dilemma. Superrational players either both cooperate or both deviate. In a Hofstadter equilibrium, players both cooperate.}
\label{fig-prisoner-dilemma-hofstadter}
\end{figure}

It is crucial to understand that the interdependence of the decisions that leads to both players cooperating in the prisoner's dilemma is not built on any supernatural telepathic powers. It is solely based on the simultaneous use of the same mathematical and logical laws on the same shared set of assumptions.

\subsection{Going beyond symmetric games}

The goal of this paper is to extend superrational reasoning to all games in normal form, with the only restriction that there are no ties in the payoffs, that is, players are never indifferent between any two outcomes\footnote{Also called general positions in the game theory literature.}. In practice, games with ties can be turned into games without ties by adding a small noise to the payoffs.

The generalized equilibrium is called the Perfectly Transparent Equilibrium (PTE). We show that, while it does not always exist, when it does exist, it is always unique, is always Pareto-optimal, and coincides with Hofstadter's equilibrium on symmetric games.

Two concepts are introduced in the reasoning underlying the PTE: 1. Necessary Rationality, i.e., rationality in all possible worlds, and 2. Necessary Knowledge of Strategies, i.e., the agents correctly predict each other in all possible worlds. The Nashian hypothesis, which states that deviations of strategies are unilateral, however, is dropped.

Necessary Rationality and Necessary Knowledge of Strategies thus provide a more generic and fine-grained epistemic support for superrationality (Section \ref{section-superrationality}) than its original characterization by Hofstadter as a direct maximization of payoffs on the diagonal. These concepts were first described in philosophical papers \citep{Dupuy1992} \citep{Dupuy2000}\footnote{Dupuy uses the term ``essential prediction.''} and applied to games in extensive form with perfect information by \cite{Fourny2018}. However, we assume that the reader did not read these papers and explain these concepts in detail in this paper. 

We finish the paper with examples, counter-examples, a review of non-Nashian literature as well as a few inclusion theorems on the specific case of symmetric games that connect the PTE (Necessary Rationality), Hofstadter's equilibrium (Superrationality) \citep{Hofstadter1983}, Halpern's and Pass's minimax rationalizability (Common Counterfactual Belief of Rationality) \citep{Halpern:2013aa}, and individual rationality as found in the commonly known Folk theorems.

\section{Background: Nashian game theory}
\label{section-normal-form}

Before we introduce our paradigm and the equilibrium, we start with some mainstream background in game theory. We start with the definitions of games in normal form as well as the Nash equilibrium, which is based on unilateral deviations. This section can be safely skipped by readers familiar with game theory (normal form, Nash, individual rationality).

\subsection{Games in normal form}

In game theory, games are typically expressed in two forms: normal and extensive. In extensive form, the game is expressed as a tree. At each node, a player picks a child node, and the leaves describe possible outcomes and are labeled with payoffs. Such games are discussed by \cite{Fourny2018} and are thus omitted in this paper.

On the other hand, in normal form, the focus of this paper, time plays no role and the payoffs are organized in a matrix (or tensor if there are more than two players). Figure \ref{fig-normal-form} shows a two-player game in normal form. One player plays on the row, the other plays on the columns.

\begin{figure}
\begin{center}
\begin{tabular}{|r|c|c|}
\hline
& C & D\\
\hline
A & $u_1(A, C), u_2(A, C)$ & $u_1(A, D), u_2(A, D)$\\
\hline
B & $u_1(B, C), u_2(B, C)$ & $u_1(B, D), u_2(B, D)$\\
\hline
\end{tabular}
\end{center}
\caption{A game in normal form, with two players that each can pick two strategies (A and B for the row player, C and D for the column player)}
\label{fig-normal-form}
\end{figure}

\begin{definition}[Game in normal form]

A game in normal form is defined with:

\begin{itemize}
\item a finite set of players $P$.
\item a set of strategies $\Sigma_i$ for each player $i\in P$.
\item a specification of payoffs $u_i(\overrightarrow\sigma)$ for each player $i\in P$ and strategy profile $\overrightarrow\sigma=(\sigma_j)_{j\in P}$, where $\forall i \in P, \sigma_i \in \Sigma_i$.
\end{itemize}

\end{definition}

\begin{figure}
\begin{center}
\begin{tabular}{|r|c|c|}
\hline
& Defect & Cooperate\\
\hline
Defect & \cellcolor{black!75}\textcolor{white}{1, 1} & 3, 0\\
\hline
Cooperate & 0, 3 & \cellcolor{black!25} 2, 2\\
\hline
\end{tabular}
\end{center}
\caption{The prisoner's dilemma, the most known game in normal form. A player can either cooperate or defect. If both cooperate, they get more than if both defect. However, a player who unilaterally defects will get even more payoff than with mutual cooperation. In the Nash equilibrium, both players defect as these strategies are each other's best responses. The individually rational outcomes are all those that weakly Pareto-dominate the maximin tuple (1,1), that is, those on the diagonal.}
\label{fig-prisoner-dilemma}
\end{figure}

\begin{figure}
\begin{center}

\begin{tabular}{|r|c|c|c|}
\hline
& A & B & C \\
\hline
D &  \cellcolor{black!25} 6, 8& 1, 2 & 4, 4\\
\hline
E &3, 1& 0, 7&  2, 3\\
\hline
F & 7, 0&  \cellcolor{black!25} 8, 5&  \cellcolor{black!75}\textcolor{white}{5, 6}\\
\hline
\end{tabular}
\end{center}
\caption{A larger game, asymmetric and with general positions. The Nash equilibrium is CF, i.e. (5, 6). Indeed, C is the best response to F ($6>0$ and $6>5$) and F is the best response to C ($5>4$ and $5>2$). Other individually rational outcomes in addition to CF are AD (6,8) and BF (8,5) because, with any other outcome, the row player would deviate to F to secure a minimum payoff of 5, or the column player would deviate to C to secure a minimum payoff of 3.}
\label{figure-asymmetric-game}
\end{figure}

The payoff space only needs to be totally ordered (preference relation). In particular, when numbers are used, they are only meant as ordinals. Comparing 1 to 1000 is no different from comparing 1 to 2. This is why in all our examples we use an increasing sequence of small integers.

Very often, given a player $i$, we denote as $\Sigma_{-i}$ the cartesian product of the remaining strategy spaces, and given a strategy profile $\overrightarrow\sigma$, we denote as $\sigma_{-i}$ the projection of the profile on $\Sigma_{-i}$. This conveniently allows writing $u_i(\overrightarrow\sigma)$ as $u_i(\sigma_i, \sigma_{-i})$, a slight abuse of notation that is always clear from the context.

We will only consider pure strategies, meaning that players may not use randomness to build mixed strategies. The outcome of a game must thus always be one of the strategy profiles of the normal form matrix, with each player getting the corresponding payoff. Furthermore, we assume that there are no ties, meaning that a player always has a strict preference between any two strategy profiles. The players do not cooperate. They act selfishly but accept and use the laws of logic. They are also rational in the sense that they make decisions to optimize their utility to the best of their knowledge or beliefs\footnote{Our paradigm differs from Nash in terms of knowledge or beliefs regarding the underlying assumptions. People can act rationally in different ways if they have different beliefs.}. They commonly know the structure of the game. Finally, the game is only played once: this is not a repeated-game equilibrium.

\subsection{Nash equilibria}

Nash equilibria are defined having in mind that players hold their opponent's choices of strategies as fixed. Restating that in subjunctive tense\footnote{As we will see in Section \ref{section-newcomb}, this expresses a counterfactual dependence.}, if a player \emph{had picked} a different strategy, the other players' strategies \emph{would still have been} the same. Payoffs are thus compared across rows or columns. With this Nashian mindset, the definition of a Nash equilibrium naturally arises as a strategy profile $\overrightarrow\sigma$ for which, for each player, the picked strategy $\sigma_i$ is the best response to the other players' strategies $\sigma_{-i}$. Formally:

\begin{definition}[Nash equilibrium] Given a game $(P, \Sigma, u)$, a strategy profile $\overrightarrow\sigma$ is a Nash equilibrium if, for any player $i\in P$:

$$\forall \tau_i \in \Sigma_i, u_i(\sigma_i, \sigma_{-i}) \ge u_i(\tau_i, \sigma_{-i})$$
\end{definition}

Figures \ref{fig-prisoner-dilemma} and \ref{figure-asymmetric-game} show the Nash equilibria for two games. In the context of this paper, the most crucial part to understand is that only unilateral deviations are considered. This leads in particular to dominant strategies, such as defecting in the prisoner's dilemma, as 1 is compared to 0 (in the same row or column) and 3 is compared to 2 (in the same row or column).

\subsection{Individual rationality}
\label{section-individual-rationality}

A broader (meaning less restrictive) concept than the Nash equilibrium is that of individual rationality. Individual rationality has been known for a few decades in the context of a few theorems on repeated games commonly known as Folk theorems\footnote{As a consequence, this solution concept cannot be attributed to a specific person.}.

It is easier to define individual rationality negatively. Let us assume that the strategies picked by all agents are known to all of them, and thus the complete strategy profile (outcome) reached by the game is known as well to all agents. If, for one of the players (say the row player), her payoff with this strategy profile is worse for her than \emph{all} of the payoffs on a different row, then this outcome is said not to be individually rational. Indeed, this player could have selected the strategy corresponding to this other row and would then have secured, with a 100\% guarantee, a better payoff, so her initial choice was irrational. A positive (equivalent, but more formal) definition of individual rationality is that a strategy profile is individually rational if it Pareto-dominates a virtual strategy profile made of all ``best worst payoffs'', that is, each player gets at least what they have the power to guarantee themselves by picking the strategy with the highest worst payoff, regardless of what the opponents do. Formally:

\begin{definition}[individually rational strategy profile] a strategy profile $\overrightarrow\sigma$ is individually rational if

$$\forall i \in P, u_i(\overrightarrow\sigma) \ge \max_{\tau_i\in\Sigma_i} \min_{\tau_{-i}\in\Sigma_{-i}} u_i(\overrightarrow\tau)$$

\end{definition}

A Nash equilibrium is always individually rational. This follows from each agent's strategy being optimal given the other agent's strategies. In the prisoner's dilemma (Figure \ref{fig-prisoner-dilemma}), the individually rational strategy profiles are exactly those on the diagonal. Indeed, each player, by defecting, has a guaranteed worst payoff of 1, so that any strategy profile with a 0 payoff is not individually rational. Figure \ref{figure-asymmetric-game} shows the individually rational strategies for the asymmetric game. The maximin of the row player is 5 because strategy F guarantees him a payoff of 5 no matter what. The maximin of the column player is 3 because strategy C guarantees her a payoff of 3 no matter what. Thus, any strategy profile that has at least a player's payoff strictly smaller than their maximin is not individually rational, as this player would otherwise deviate to C, respectively to F. Three outcomes are individually rational: AD, BF, and CF.


\section{Background: Newcomb's problem and the non-Nashian approach}
\label{section-newcomb}

We now introduce Newcomb's problem, based on which we discuss Causal Decision Theory (CDT), Evidential Decision Theory (EDT), and a proposed ``third way out'', which we call Non-Nashian Decision Theory (NNDT)\footnote{The philosophical seed of this alternative decision theory and a deeper discussion of rational reasoning with Newcomb's problem is found in \citep{Dupuy1992}. Dupuy suggested the terminology ``Counterfactual Decision Theory'' in his seeding paper, however we prefer to call it Non-Nashian Decision Theory to avoid a collision in the abbreviations, but also to clarify that counterfactuals also appear in CDT, even though they are constrained there by independence assumptions.}. We also give a short introduction to the basics of possible worlds semantics.

Readers familiar with Newcomb's problem, CDT, EDT, possible worlds (Kripke), counterfactual dependencies and Dupuy's seeding work may directly skip to Sections \ref{section-nnpt} and \ref{section-pte} where we introduce our new solution concept. Readers who only need to catch up with counterfactual thinking can skip to \ref{section-nndt}, where we introduce NNDT, knowing that our notation is documented in Section \ref{section-event-notations}.

\subsection{Newcomb's problem}

Game theory involves anticipating the decisions of players endowed with free choice. The apparent conflict between the two concepts is embodied in Newcomb's problem\footnote{Newcomb's problem is also called Newcomb's paradox, however, as explained below, we view it less a paradox than as a thought experiment to illustrate the importance of making counterfactual dependencies explicit.}, which is as follows. An agent is facing two boxes, one is opaque, and one is transparent. There are \$1,000 in the transparent box. The agent knows the opaque box is either empty or contains \$1,000,000. The agent may either pick the opaque box or both boxes. The catch is that, previously, an entity predicted what the agent would do: if she predicted the agent would pick one box, she put \$1,000,000 inside that opaque box. If she predicted the agent would pick two boxes, she put nothing inside. It is further known that the predictor has done this thousands of times, and all her predictions were correct hitherto.

This problem is often referred to as a paradox because two lines of reasoning seem equally reasonable: a Nashian, dominant strategy argument comparing \$x to \$x+1,000 leads to two boxes being optimal no matter what the value of $x$ is. But another reasoning taking for granted that the prediction is always correct compares \$1,000 to \$1,000,000 and leads to one box being optimal.

\subsection{Several decision theories}

Newcomb's problem is often explained in light of Causal Decision Theory (CDT) and Evidential Decision Theory (EDT), which maximize payoffs (utility or value) in different ways \citep{Weirich2016}. \citet{Gibbard1978} showed that the difference comes down to distinguishing between probabilities of subjunctive conditionals and conditional probabilities. CDT, built on probabilities of subjunctive conditionals, supports two-boxers, while EDT, built on conditional probabilities, supports one-boxers.

The framework on which this paper is based is a third way out (Non-Nashian Decision Theory, NNDT) that is neither CDT nor EDT and which supports one-boxers. In order to formally differentiate it from CDT and EDT, we first introduce CDT and EDT on the example of the Newcomb problem.

\subsection{Possible worlds}
\label{section-possible-worlds}

Before we discuss the three decision theories, we say a few words about possible worlds, which all approaches have in common underneath. Possible worlds date back to as early as Gottfried \citet{Leibniz1710}, also with the notion of maximization of a quantity over the set of possible worlds. \cite{Kripke1963} later designed a formal epistemic framework to model epistemic logic statements. There are a lot of variants of such semantics, but we give here a short ``Possible worlds 101'' that will help a larger audience understand counterfactuals as well as the differences between the three decision theories.

We have a set $\Omega$ of possible worlds, typically denoted $w\in\Omega$. Agents make choices, and different choices are made in different worlds. Formally, in each world, for each agent who has to make a choice, a specific decision is made. Regarding Newcomb's problem, in some worlds, the agent picks one box; in some others, the agent picks two boxes. Likewise, in some worlds, the prediction is ``one box,'' and in some other worlds, the prediction is ``two boxes''. In each world, the choice and the prediction are unambiguously defined. Events are subsets of $\Omega$. For example, ``the agent picks one box'' is an event that corresponds to the subset of all worlds in which the agent picks one box. It can, equivalently, be seen as a truth assignment (true or false) on $\Omega$.

Knowledge, and often actually, incomplete knowledge, is formalized with associating each world and agent to some subset of $\Omega$. This relation is called the accessibility relation. Saying that for some world $w$ and agent $i$, a set $A$ of worlds is accessible, models the fact that agent $i$ in world $w$ knows that they are in some world in $A$, but do not know which one exactly (so they may not know that this is, actually, $w$). Typically, $w\in A$. If we allow for cases where $w\notin A$, we prefer the word ``belief'' (doxastic) to ``knowledge'' (epistemic) because the agent may have incorrect beliefs. Very often, but not always, the accessibility relation is simply expressed as a partition of $\Omega$, which is a special case\footnote{when the relation is transitive, symmetric, and reflexive.}. The said agent knows or believes in an event $B$ if $A\subseteq B$ (modeling knowledge of events in this way is probably the most brilliant insight of Kripke).

The necessity of an event $A$ means that $A=\Omega$. This is why necessary rationality is the same as rationality in all possible worlds, etc. Perfect prediction, in this paper, is simply defined as the fact that accessible worlds from $w$ for any agents are the singleton $\{w\}$, and a full formal account is given in a technical report \citep{Fourny2018b}. In a Bayesian approach, given some actual world, a probability measure is introduced on accessible worlds (say, $A$). These probabilities model the relative (un)certainty of the agent in their knowledge. Events can then be reinterpreted as random variables, and we can also define conditional probabilities and compute correlations, which is always symmetric. This is what EDT does. EDT limits itself to worlds that are accessible from the actual world and does not look beyond those.

In an approach with counterfactuals, a different relation on events, which we denote here $>$, is introduced instead of correlations, and that relation can be asymmetric. Given two events $A$ and $B$, $A> B$ is a subjunctive conditional that says that ``if A were true, then B would be true." The meaning of counterfactuals has been discussed by \citet{Stalnaker1968} and formalized by \cite{Lewis1973} in terms of possible worlds. Lewis suggested organizing alternative possible worlds (taken from $\Omega$) around the actual world, with a notion of distance. Then, the counterfactual statement, or subjunctive conditional,
$$A > B$$
is true in some world $w$ if, in the closest world to $w$ in which A is true, denoted $f(w, A)$, B is also true. $A>B$ is thus, formally, also an event, like its operands A or B. This means that counterfactuals nest recursively, and we can also write $f(w, A>B)$, $f(f(w, A), B)$, etc.

The term ``counterfactual'' comes from the fact that $f(w, A)$ may (but need not) be an inaccessible world (contrary to the facts known in $w$). The approach with counterfactuals is taken both by CDT and NNDT, and the difference lies in whether the past is (CDT), or does not have to be (NNDT), counterfactually independent of an agent's decision. Within CDT and NNDT, we can also define probabilities on subjunctive conditionals, i.e., they are random variables, too, but it is crucial to see that their nature is fundamentally different from (always symmetric) conditional probabilities. A very simple counterexample is that subjunctive conditionals do not generally follow the total law of probability\footnote{This is well known to physicists specialized in quantum foundations (contextuality). This law is replaced with the Born rule.}.

With this in mind, let us now interpret Newcomb's problem and show that, in all three cases, the agent is maximizing their utility.

\subsection{Event notations for Newcomb's problem}
\label{section-event-notations}

In the following, we consider the following events, subject to some underlying probability distribution:  $\blacksquare$ is the event that the opaque box has a million dollars inside. $\square$ is the event that the opaque box is empty. $ONE$ is the event that one picks one box. $TWO$ is the event that one picks two boxes.

Furthermore, we have a random variable P that is the choosing agent's utility (in dollars). This variable has an expected value if the agent picks one box, $\mathbb{E}_{ONE}[P]$, and an expected value if the agent picks two boxes, $\mathbb{E}_{TWO}[P]$. All decision theory frameworks have in common that utility is maximized. They differ in the definition of $\mathbb{E}_{ONE}[P]$ (the expected utility if one picks one box) and $\mathbb{E}_{TWO}[P]$ (the expected utility if one picks two boxes) out of the four possibilities:

$$\mathbb{E}_{TWO,\blacksquare}[P] = 1,001,000$$
$$\mathbb{E}_{ONE,\blacksquare}[P] = 1,000,000,$$
$$\mathbb{E}_{TWO,\square}[P] = 1,000,$$
$$\mathbb{E}_{ONE,\square}[P] = 0$$

\subsection{Causal Decision Theory}

In CDT, the impact of the decision on the environment, and in turn on the payoffs, is modeled with probabilities of subjunctive conditionals. $ONE > \blacksquare$ is the event that ``If I picked one box, then the opaque box would have a million dollars." $ONE > \square$ is the event that ``If I picked one box, then the opaque box would be empty." In CDT, expected utility is computed like so:

$$\mathbb{E}^{CDT}_{ONE}[P] = \mathbb{E}_{ONE,\blacksquare}[P] * P(ONE > \blacksquare ) + \mathbb{E}_{ONE,\square}[P] * P(ONE> \square)$$
and
$$\mathbb{E}^{CDT}_{TWO}[P] = \mathbb{E}_{TWO,\blacksquare}[P] * P(TWO > \blacksquare ) + \mathbb{E}_{TWO,\square}[P] * P(TWO> \square)$$

There is another important assumption commonly made in Causal Decision Theory: namely, that if an agent makes a free decision, it is counterfactually independent of anything not in its future. Since the decision cannot cause its anticipation, it follows, according to Causal Decision Theorists, that the anticipation is counterfactually independent of the decision. \citet{Stalnaker1972} thus argues that the probability of the subjunctive conditional ``were one to take one box, the prediction would have been one-box'' should be the same as the probability that the prediction is one-box.

In other words, since the decision in Newcomb's problem is in the future of when the box is filled,
$$P(ONE > \blacksquare)=P(\blacksquare)$$
and
$$P(ONE > \square)=P(\square)$$

In other words, the decision has no impact on the content of the opaque box. As a consequence, the expectations are calculated like so:

\begin{align*}
\mathbb{E}^{CDT}_{ONE}[P] &= 1000000 * P(ONE > \blacksquare ) + 0 * P(ONE> \square)\\
& = 1000000 * P(\blacksquare )+ 0 * P(\square)\\
& = 1000000 * P(\blacksquare )
\end{align*}

\begin{align*}
\mathbb{E}^{CDT}_{TWO}[P] &= 1001000 * P(TWO > \blacksquare ) + 1000 * P(TWO > \square)\\
&= 1001000 * P(\blacksquare ) + 1000* P(\square)\\
&= 1000000 * P(\blacksquare ) + 1000* (P(\blacksquare )+P(\square))\\
&= \mathbb{E}^{CDT}_{ONE}[P]  + 1000\\
&> \mathbb{E}^{CDT}_{ONE}[P]
\end{align*}

which leads to the decision of picking two boxes, which has the higher utility.

\subsection{Evidential Decision Theory}

Evidential decision theorists, on the other hand, compute the expected utility with conditional probabilities involving the events $ONE$, $TWO$, $\blacksquare$ and $\square$. Conditional probabilities behave differently than subjunctive conditionals. 

In the EDT framework, the expected utilities for the two possible choices are calculated with:

$$\mathbb{E}^{EDT}_{ONE}[P] = \mathbb{E}_{ONE,\blacksquare}[P] * P(\blacksquare | ONE) +  \mathbb{E}_{ONE,\square}[P] * P(\square | ONE) )$$
and
$$\mathbb{E}^{EDT}_{ONE}[P] = \mathbb{E}_{TWO,\blacksquare}[P] * P(\blacksquare | TWO) +  \mathbb{E}_{TWO,\square}[P] * P(\square | TWO) )$$

where it is assumed from the problem formulation that

$$P(\blacksquare | ONE)=P(\square | TWO) = 1$$
and
$$P(\blacksquare | TWO)=P(\square | ONE) = 0$$

As a consequence:

\begin{align*}
\mathbb{E}^{EDT}_{TWO}[P] &= 1001000 * P(\blacksquare | TWO) + 1000 * P(\square | TWO) )\\
&= 1001000 * 0+ 1000 * 1\\
&= 1000
\end{align*}

\begin{align*}
\mathbb{E}^{EDT}_{ONE}[P] &= 1000000 * P(\blacksquare |ONE) + 0 * P(\square | ONE) )\\
& = 1000000 * 1 + 0 * 0\\
&= 1000000 > \mathbb{E}^{EDT}_{TWO}[P]
\end{align*}

and thus it is rational to pick one box.

%
%
%
%
%
%
%
%

\subsection{A third way out: Non-Nashian Decision Theory}
\label{section-nndt}

The discussion between CDT and EDT is the subject of intense debate. But we argue here that there is more to this.

Causal Decision Theory is not only characterized by the use of probability of subjunctive conditionals. A paramount assumption in CDT is also the strong, Nashian free-choice assumption that the prediction of a decision is counterfactually independent of this decision. \citet{Dupuy1992} argues that it is crucial to distinguish causal from counterfactual independence. Indeed, the absence of causation does not, in general, and also according to the laws of physics, imply counterfactual independence. Measuring entangled particles in quantum systems illustrates this, as it can be expressed by subjunctive conditionals: if the measurement on the one side had been different, then the measurement on the other side would also have been counterfactually different, even though no causal effect can apply.

Using the probability of the subjunctive conditionals as in CDT, but replacing the free-choice assumption with a counterfactual dependence between the decision and its prediction yields an alternative theory of rational choice. We call this third theory Non-Nashian Decision Theory (NNDT). In NNDT, we do use, as in CDT, subjunctive conditionals, but if we drop the free-choice assumption, and assume instead that the prediction is counterfactually dependent on the decision, the following subjunctive conditionals are always true:

``if one had picked one box, the prediction would have been one-box''

``if one had picked two boxes, the prediction would have been two-boxes''

so that we have

$$P(ONE > \blacksquare)=P(TWO>\square)=1$$

and

$$P(TWO > \blacksquare)=P(ONE>\square)=0$$

And as a consequence:

\begin{align*}
\mathbb{E}^{NNDT}_{TWO}[P] &= 1001000 * P(TWO > \blacksquare) + 1000 * P(TWO > \square) )\\
&= 1001000 * 0+ 1000 * 1\\
&= 1000
\end{align*}

\begin{align*}
\mathbb{E}^{NNDT}_{ONE}[P] &= 1000000 * P(ONE > \blacksquare) + 0 * P(ONE > \square) )\\
& = 1000000 * 1 + 0 * 0\\
&= 1000000 > \mathbb{E}^{NNDT}_{TWO}[P]
\end{align*}

and the expected utility of picking one box outweighs that of picking two. 

Although the calculation, in the very case of the Newcomb problem, is similar and yields the same result than EDT, it is crucial to understand that the underlying reasoning is very different. Picking one box, in the NNDT mindset, is not merely evidence that the prediction was ``one box''. Picking one box has a counterfactual impact on the prediction, something formally modeled with Lewisian counterfactual functions. To understand this, it helps to consider that the agent may not be able to distinguish whether they are actually making their decision, or whether they are being simulated by the predictor.

As explained in \citet{Fourny2018}, Newcomb's problem is underspecified in the sense that this Lewisian counterfactual function corresponding to changing one's choice is not constrained by the problem formulation. In a CDT mindset, the counterfactual function expresses that the prediction would have been the same if one had decided otherwise. In Nashian game theory, this translates to unilateral deviations of strategies. In the NNDT mindset, the counterfactual function expresses that the prediction \emph{would also have been correct} if one had decided otherwise. This is a weaker form of free choice in which the agent \emph{could have acted otherwise}, but if they had, then the prediction would have been different. In the EDT mindset, there are no counterfactual functions at all, only conditional probabilities. It is important to understand that NNDT is distinct from EDT

\citet{Gibbard1978}, \citet{Lewis1979} and \citet{Dupuy1992} argued that there is a direct analogy between Newcomb's problem and the prisoner's dilemma: in both problems, the question is whether the decisions are counterfactually interdependent (NNDT), counterfactually independent (CDT), or statistically correlated (EDT).

\section{Necessary Rationality and Necessary Knowledge of Strategies in game theory}
\label{section-nnpt}

Now, let us come back to game theory and illustrate our new solution concept on the Prisoner's dilemma. The general formalism will then be introduced in Section \ref{section-pte}.

\subsection{CDT and Nash equilibria}

CDT and NNDT both rely on subjunctive counterfactuals modeling the impact of decisions on utilities but fundamentally differ, namely: whether or not a decision is counterfactually independent of anything it could not have caused\footnote{The latter part of this formulation of the strong version of free choice is attributable to theoretical physicists \citep{Renner2011}, even though they refer in this paper to conditional probabilities and the absence of correlations.}.

The Nash paradigm directly follows the CDT paradigm. Let us take the example of games in normal form, in which players may be, for example, in separate rooms to pick their strategies. If an agent's decision is counterfactually independent of anything it could not have caused, then this means that if some agent Mary picks strategy $\sigma_i$ and the other agents picked strategy $\sigma_{-i}$, the following counterfactual statement holds: ``Had Mary picked another strategy $\tau_i$, the other agents would still have picked strategies $\sigma_{-i}$.'' In Lewisian terms: in the closest world in which Mary picks $\tau_i$, the other agents jointly pick $\sigma_{-i}$ like in the actual world. Mary's decision has no impact on the other agents' strategies.

This counterfactual statement is what is known as unilateral deviations. It is thus rational for Mary to pick $\sigma_i$ over $\tau_i$ if\footnote{The inequality is always strict because of the assumption of general positions.}

$$\forall \tau_i \neq \sigma_i, u_i(\sigma_i, \sigma_{-i}) > u_i(\tau_i, \sigma_{-i})$$

In Lewisian semantics: Mary is rational because she gets a higher utility in the actual world than in the closest world in which she picks $\tau_i$, for any $\tau_i\neq \sigma_i$. She \emph{would have} gotten less utility if she \emph{had} picked another strategy. This is why the Nash equilibrium is computed by finding best responses to the opponents' strategies, by fixing a row or column and maximizing utility over it. The Nash equilibrium is reached under this strong free choice assumption and assuming Common Knowledge of Rationality in the actual world.

\subsection{Perfect Prediction and non-Nashian thinking: Necessary Knowledge of Strategies and Necessary Rationality}

Hofstadter was the first, to our knowledge, to drop the strong free choice assumption in a game-theoretical context. The Hofstadter equilibrium \citep{Hofstadter1983} is defined for symmetric games (see a formal definition in Section \ref{section-symmetric-games}) in which all agents have the same choice of strategies and the same payoff configurations, and limits the possible outcomes of the game to the diagonal (see a formal definition in Section \ref{section-superrationality}). While many refer to Hofstadter's idea (superrationality) in terms of EDT, we argue here that it is also meaningful, if not more meaningful, to formulate it in subjunctive conditionals -- within our third decision theory, NNDT. It is doing so that provides a way to generalize the reasoning to asymmetric games.

Formally, this is, again, supported by a counterfactual statement: ``If Mary had picked another strategy $\tau$, then all other agents would have picked that same strategy $\tau$.'' It can immediately be seen that this counterfactual statement is in direct contradiction with the CDT statement made in the former section: we are in the NNDT realm, in which deviations need not be unilateral. It is rational to cooperate for the prisoner's dilemma, because, had one defected, the opponent would have defected too.

Thus, under NNDT, utility is maximized, but on the diagonal, based on the counterfactual impact of a deviation of strategy on the other agents' strategies:

$$u_i(C, C)>u_i(D, D)$$

This is to be contrasted to the comparison that is done in CDT, leading to defection:

$$u_i(C, C)<u_i(D, C)$$

Now, let us dig into the counterfactual reasoning. First, we are going to drop the CDT free choice assumption to obtain the NNDT paradigm. Second, more specifically, we are going to replace this assumption with Perfect Prediction assumptions. This is one way to instantiate the NNDT paradigm -- there are other ways with weaker assumptions than Perfect Prediction (see \citet{Halpern:2013aa}, for example).

Perfect Prediction relies on two fundamental principles, which are embodied in Hofstadter's original statement in Section \ref{section-introduction}, and were precisely laid out by \cite{Dupuy1992} and \cite{Dupuy2000} as projected time:

\begin{itemize}
\item Principle 1: Necessary Knowledge of Strategies. The agents know each other's strategies in all possible worlds. Let us express this in counterfactual terms. All agents know, in advance, the outcome (strategy profile) that the game will reach, say $\sigma$. If the outcome of the game had been different, say $\tau$ -- that is, if an agent had acted otherwise -- then all agents \emph{would have known} that it would have been $\tau$.
\item Principle 2: Necessary Rationality. All agents are rational in all possible worlds. In counterfactual terms, if a player had picked a different strategy, all agents \emph{would still have} acted rationally. Necessary Rationality is a stronger assumption than Common Knowledge of Rationality: if the agents are rational in all possible worlds, then it follows in the actual world that the agents know, and commonly know, that they are rational\footnote{It is a very short proof that we can easily summarize here: An event -- or logical predicate -- can be canonically identified with the set of the possible worlds in which this event happens. In \cite{Kripke1963} semantics, knowledge of an event is defined as the fact that the set of all epistemically accessible worlds is included in this event. Necessary rationality means that the event that the agents are rational is the set of all possible worlds -- the inclusion thus follows trivially, and common knowledge follows (i) by repeating this argument in every possible world, which leads to necessary knowledge of rationality, and then (ii) by applying the entire argument recursively on the event of (knowledge of)$^n$ rationality for any $n$.}.
\end{itemize}

The decision making and prediction of the actual outcome rely solely on mathematical and logical computations and are a consequence of the at-most uniqueness of the outcome that can possibly be reached under this belief. Actually, the reasoning is based on \emph{reductio ad absurdum}, showing that alternative candidate worlds are, in fact, ``impossible possible'' worlds\footnote{We use this terminology, having in mind the work of \citet{Rantala1982} and \citet{Kripke1965}.}.

The formalization of a game-theoretical equilibrium in extensive form -- when agents play one after the other -- has already been published by \citet{Fourny2018} as the Perfect Prediction Equilibrium (PPE) under these same assumptions of Necessary Knowledge of Strategies and Necessary Rationality. The reasoning behind the PPE is closely related to the elapse of time: saying that the solution must be immune to its prediction amounts to say that the solution must be \emph{caused} by its prediction, like a self-fulfilling prophecy.

In this paper, we are looking at games in normal form. We call the equilibrium the Perfectly Transparent Equilibrium (PTE).
\subsection{A reformulation of the Prisoner's dilemma reasoning}
\label{section-prisoner-dilemma-pte}

Let us now revisit the prisoner's dilemma, shown in Figure \ref{fig-prisoner-dilemma} with the above principles of Necessary Knowledge of Strategies and Necessary Rationality in mind, and now departing from any diagonal argument nor using the knowledge that the game is symmetric.

We can start the reasoning, under the assumptions of Necessary Rationality and Necessary Knowledge of Strategies, by looking, for each possible strategy that a player can pick, for the worst payoff he can get. The row player gets at least zero by cooperating, and at least one by defecting. The same applies to the column player. We can thus see that defecting guarantees to both players a payoff of at least one: this is the maximin of both players.

\subsubsection{First round of elimination}

All outcomes that are not individually rational can be shown to be incompatible with Necessary Rationality and Necessary Knowledge of Strategies. They are even incompatible with the weaker assumptions of rationality (in the actual world) and knowledge of strategies (correct prediction in the actual world).

Let us show, for example, that outcome (0,3), reached when the row player cooperates and the column player defects (CD), is inconsistent with these assumptions. Let us assume the row player chooses to cooperate and knows that the outcome of the game will be CD (0,3). This means he gets a payoff of 0. Is he rational? By definition, he is if he gets a better utility with his actual choice than the utility he would have counterfactually obtained if his choice had been different. So what if he had defected? We need to look at the ``counterfactual payoff'' that he \emph{would have gotten}  if he had defected. The following statement holds directly because of the rules of the game itself (all payoffs for the row player on the D line are at least 1) -- and this is completely independent of anything the column player would have done: ``If the row player had chosen to defect, he would have obtained a payoff of at least 1.''

So if he had chosen to defect, the row player would have obtained a better payoff ($1>0$). Thus, we proved that (note the indicative tense: this is a logical implication): ``if the row player chooses to cooperate and knows that the final outcome will be CD (0,3), then his decision to cooperate is not rational.''

Likewise, for DC (3,0), we can assume the column player chooses to cooperate, knowing the outcome of the game will be DC (3,0) and that she thus gets 0. But if she had defected, she would have obtained at least 1. Thus, we showed that the following logical statement holds: ``If the column player chooses to cooperate and knows that the final outcome will be CD (0,3), then her decision to cooperate is not rational.''

This same reasoning can be done, for any game, with any outcome that is not individually rational: non-individually rational outcomes cannot be willingly reached under Necessary Rationality and Necessary Knowledge of Strategies. Since the maximin of both players is 1, any outcome with a payoff less than one for the row player, or with a payoff less than one for the column player, contradicts our assumptions. This is the case with DC and CD. The first and second matrix in Figure \ref{fig-example1} show the elimination of CD and DC on the game.

What about the other outcomes? At that first step, only assuming rationality and knowledge of strategies, we cannot do better than that. For example, (1,1) or (2,2) may or may not correspond to rational behavior of the players, which depends on the counterfactual payoffs, the payoffs that the agents \emph{would otherwise have gotten} if they had played otherwise. To tell, we need to get to step 2 of the reasoning.

The key part of the reasoning is the awareness of the fact that these two outcomes DC (3,0) and CD (0,3) cannot be the solutions of the game under our assumptions. With these outcomes discarded, we can enter further rounds of elimination with the same reasoning: step 2.

\begin{figure}
\begin{center}

\resizebox{\textwidth}{!}{
\begin{tabular}{lcccr}

\begin{tabular}{|r|c|c|}
\hline
& C & D \\
\hline
C & 2, 2& 0, 3\\
\hline
D & 3, 0& 1,1\\
\hline
\end{tabular}

&

$\Rightarrow$

&

\begin{tabular}{|r|c|c|}
\hline
& C & D \\
\hline
C & 2, 2& \cellcolor{black!25}0, 3\\
\hline
D & \cellcolor{black!25}3, 0& 1,1\\
\hline
\end{tabular}
&

$\Rightarrow$

&

\begin{tabular}{|r|c|c|}
\hline
& C & D \\
\hline
C & 2, 2& \cellcolor{black!25}0, 3\\
\hline
D & \cellcolor{black!25}3, 0& \cellcolor{black!25}1,1\\
\hline
\end{tabular}

\end{tabular}
}
\end{center}
\caption{Iterated elimination of preempted strategy profiles in the Prisoner's dilemma. (2,2) is the unique PTE and coincides with the Hofstadter, superrational equilibrium.}
\label{fig-example1}
\end{figure}

\subsubsection{It is irrational to defect}

At step 2, we re-iterate our reasoning by looking for the maximins -- but taking into account that the players cannot knowingly and rationally play towards either DC (3,0) or CD (0,3).

The row player has a minimum guaranteed payoff of 1 if he defects, and of 2 if he cooperates (because CD has been eliminated). His maximin is thus 2. Likewise, the column player has a maximin of 2. With these maximins in mind, we can proceed with the reasoning that eliminates any outcomes with a payoff of less than 2 for the row player, or less than 2 for the row player. The only such outcome is DD (1,1). We now thus show that DD (1,1) cannot possibly be reached under Necessary Rationality and Necessary Knowledge of Strategies. For this, we show that, if all of the following holds:

\begin{enumerate}
\item the row player chooses to defect and knows that the final outcome will be DD (1,1) (principle 2);
\item the row player would have acted rationally if he had cooperated (principle 1);
\item the row player would also have known the outcome if it had been different (principle 2).
\end{enumerate}

then the row player is not acting rationally.

So let us assume that the row player defects, and that he knows the final outcome will be DD and thus gets a payoff of 1. To tell whether he is rational, we need to know what payoff he would have gotten if he had cooperated. From point 3 above, we know that if he had cooperated, and thus reached a different outcome, he would have known this alternative outcome. We also established, in the previous section, that if, in some possible world, the row player chooses to cooperate and knows that the final outcome will be CD (0,3), then his decision to cooperate is not rational.

It follows that, if the row player had chosen to cooperate, the final outcome could not have been CD, as (i) knowing this, he would then have been irrational, and (ii) point 2 above says that he had chosen to cooperate, he would have been rational. Thus, if the row player had cooperated, the final outcome would have been CC (2, 2), and he would have obtained a payoff of 2. Since $2>1$, it follows directly that the row player, choosing to defect and knowing that if he had cooperated, he would have gotten a higher payoff, is irrational. (1,1) is thus also incompatible with Necessary Rationality and Necessary Knowledge of Strategies.

\subsubsection{It is rational to cooperate}

In the Nash paradigm, we compare payoffs to find a best response to the opponent's strategy, held fixed, and this best response is the rational choice. But in the non-Nashian paradigm, knowing that it is irrational to defect does not imply that it is rational to cooperate, because utilities are compared against alternative possible worlds, and counterfactual dependencies are not a symmetric relationship in general. We need to show that it is rational to cooperate with separate reasoning.

For this, we show that, if all of the following holds:

\begin{enumerate}
\item the row player chooses to cooperate and knows that the final outcome will be CC (2,2) (principle 2);
\item the column player\footnote{We do mean the column player -- this shows how this differs from the reasoning in the previous section because counterfactual dependencies are asymmetric.} would have acted rationally if the row player had defected (principle 1);
\item the column player would also have known the outcome if it had been different (principle 2).
\end{enumerate}

then the row player is acting rationally. 

So let us assume that the row player cooperates, and that he knows the final outcome will be CC and thus gets a payoff of 2. To tell whether he is rational, we need to know what payoff he would have gotten if he had defected.

From point 3 above, we know that if he had defected, and thus reached a different outcome, the column player would have known this alternative outcome. We also established, in the previous section, that if, in some possible world, the column player chooses to cooperate and knows that the final outcome will be DC (3, 0), then her decision to cooperate is not rational.

It follows that, if the row player had chosen to defect, the final outcome could not have been DC, as (i) knowing this, the column player would then have been irrational, and (ii) point 2 above says that he had chosen to defect, she -- the column player -- would have been rational. Thus, if the row player had defected, the final outcome would have been DD (1, 1), and he would have obtained a payoff of 1. Since $2>1$, it followed directly that the row player, choosing to cooperate and knowing that if he had defected, he would have gotten a lesser payoff, is acting rationally.

The same reasoning can be done for the column player, for whom it is rational to cooperate under our assumptions. The only remaining outcome, (2,2), thus correspond to the rational choice of cooperating for both players, under Necessary Rationality and Necessary Knowledge of Strategies. This is shown in the final matrix of Figure \ref{fig-example1}.

The key difference with the Nash equilibrium is in the counterfactual statements, which lead to a different comparison of the utilities obtained with each strategic choice and under counterfactual implications. Under Necessary Rationality and Necessary Knowledge of Strategies (or the belief thereof), it is rational for both players to cooperate. This line of reasoning reaches the same outcome as Douglas Hofstadter's reasoning but does not use the argument of symmetry. It can thus be extended to other games in normal form. This is what we are going to do now.

\section{The Perfectly Transparent Equilibrium}
\label{section-pte}

The PTE is defined for games in normal form, with one additional assumption similar to that of its extensive form counterpart, namely, that the payoffs are in general position. In this part, we give a general, formal algorithm that reproduces the reasoning made in Section \ref{section-nnpt}. The epistemic proof based on Kripke semantics is given as a separate technical report \citep{Fourny2018b} and is too long (32 pages) to be included in this paper.

\subsection{Preemption}

The PTE is based on an iterated elimination of strategy profiles that cannot possibly be the solution of the game under Necessary Knowledge of Strategies and Necessary Rationality, because players otherwise \emph{would have deviated}. The elimination of such strategy profiles is called preemption.

Informally, a strategy profile $\overrightarrow\tau$ is preempted by a strategy $\sigma_i$ (of any player $i$) if player $i$ is worse off with $\overrightarrow\tau$ than with the minimum payoff that she is assured of getting $\sigma_i$ no matter what the opponents would (counterfactually) have done had she picked $\sigma_i$. It would thus be irrational to a player to pick $\tau_i$ under our assumptions. In the prisoner's dilemma, for example, at step 1, CD is preempted by the row player defecting, and DC is preempted by the column player defecting.

\subsubsection{The first round of elimination: individual rationality}

The complete elimination scheme is obtained by finding the maximin utility. The concept of maximin utility is commonly found in the game theory literature and is used in the definition of individual rationality. Each strategy has a minimum guaranteed payoff (no matter what the opponents would do), and the maximin utility for a given player is the maximal utility among these minimum guaranteed payoffs.

All strategy profiles that do not Pareto dominate the maximin utility, that is, that are not individually rational, are preempted. Indeed, if a strategy profile gives, to some player, a payoff inferior to their maximin utility, it means that the player, who would be worse off with this strategy profile, could simply have deviated to the strategy that guarantees him his maximin utility. He would thus not have been rational.

Necessary Rationality and Necessary Knowledge of Strategies\footnote{As illustrated in the former section, for the first step, rationality and knowledge of strategies in the actual world are sufficient conditions.} together entail that such outcomes that are not individually rational are thus impossible. In other words, the first round of elimination comes down to eliminating all strategy profiles that are not individually rational.

\begin{definition}[$1^{st}$-level-preempted strategy profile]
Given a game in normal form, with pure strategies and with no ties $\Gamma=(P, (\Sigma_i)_i, (u_i)_i)$, a strategy profile is $1^{st}$-level-preempted if it is not individually rational. In other words, any strategy profile that does not Pareto-dominate the maximin utility is $1^{st}$-level-preempted. For a game $\Gamma$, the strategy profiles $\overrightarrow\sigma\in \mathcal{S}_1(\Gamma)$ that survive the first round of elimination are characterized with:

$$\mathcal{S}_1(\Gamma) = \{ \overrightarrow\sigma \quad | \quad \forall i \in P, u_i(\overrightarrow\sigma) \ge \max_{\tau_i\in\Sigma_i} \min_{\tau_{-i}\in\Sigma_{-i}} u_i(\tau_i, \tau_{-i}) \}$$

\end{definition}

In the literature, the maximin utility is often called minimax, which deserves a clarification to avoid any confusion. Depending on the context, this is on the one hand because a maximin on the gains is equivalent to minimax on the losses, i.e., a player minimizes the worst-case loss. Some frameworks are based on losses rather than gains. This is also, on the other hand, because in zero-sum games, a maximin on an agent's gains is equivalent to a minimax on the opponent's gains (and to a maximin on the opponent's losses). We stick to the maximin terminology because we are looking at gains, and the formula reads as ``maximin''. Our definition of individual rationality matches that given by \citet{Halpern:2013aa}.

\subsubsection{Subsequent rounds of elimination}

In subsequent rounds of eliminations, only strategy profiles that survived previous rounds can be considered for computing the minimums in the maximin utility. Indeed, in all possible worlds, both agents know the laws of logic, and under Necessary Rationality and Necessary Knowledge of Strategies, came to the conclusion that eliminated strategy profiles are impossible\footnote{In the formal sense of impossible possible worlds, see \citep{Rantala1982}. In these worlds, truth is assigned manually to logical formulas, and the assignment may contradict logical axioms and rules. More details are found in \citet{Fourny2018b}.}. This is the distinctive feature of Perfect Prediction and of the PTE.

Also, a strategy cannot preempt in subsequent rounds if all the strategy profiles it contains have been eliminated in previous rounds -- Indeed, if all strategy profiles under a strategy are known to be impossible under Necessary Rationality and Necessary Knowledge of Strategies, then the strategy itself is impossible\footnote{It does not entail that the agent has no free will. Picking this strategy would rather be an indication that the agent does not act rationally, or does not act to the best of their knowledge, under the assumptions at hand. Necessary Rationality and Necessary Knowledge of Strategies would then not hold. This is, of course, a different kind of rationality, as an irrational strategy according to the PTE may be rational under Nash semantics, and vice-versa. In both cases, agents optimize their utility to the best of their knowledge, but the counterfactual knowledge assumptions differ between the semantics.}. It can thus no longer be considered in the computation of the maximum in the maximin utility\footnote{This must be made explicit in the equation, as in the extensive form, as we will shortly see. This is because otherwise, the minimum utility for a strategy in which all strategy profiles have been eliminated would be $+\infty$, and the maximin utility would then be $+\infty$.}.

We now give the definition of subsequent rounds, when we no longer consider previously eliminated strategy profiles.

\begin{definition}[$k^{th}$-level-preempted strategy profile]
Given a game in normal form, with pure strategies and with no ties $\Gamma=(P, (\Sigma_i)_i, (u_i)_i)$, a strategy profile is $k^{th}$-level-preempted, for $k>1$, if it does not Pareto-dominate the maximin utility, where the maximin is only taking into account strategy profiles that are not $(k-1)^{th}$-level preempted. The strategy profiles $\overrightarrow\sigma\in \mathcal{S}_k(\Gamma)$ that survived the $k^{th}$ round of elimination are characterized with:

$$\mathcal{S}_k(\Gamma) = \{ \overrightarrow\sigma | \forall i \in P, u_i(\overrightarrow\sigma) \ge$$

$$\max_{
\begin{array}{c}
\tau_i\in\Sigma_i
\\
{\scriptstyle \text{s.t.} \exists \tau_{-i}\in\Sigma_{-i}, (\tau_i, \tau_{-i})\in \mathcal{S}_{k-1}(\Gamma)}
\end{array}
}
\quad
\min_{
\begin{array}{c}
\tau_{-i}\in\Sigma_{-i}
\\
{\scriptstyle \text{s.t.} (\tau_i, \tau_{-i})\in \mathcal{S}_{k-1}(\Gamma)}
\end{array}
} u_i(\tau_i, \tau_{-i}) \}$$
\end{definition}

\subsubsection{Convergence}

The iterated elimination converges at some point.

\begin{lemma}[Convergence]
The sequence of sets $(\mathcal{S}_i(\Gamma))_{i\in\mathbb{N}}$ converges and reaches its limit: at a certain point, no more strategy profiles get eliminated. We denote this limit $\mathcal{S}(\Gamma)$
\end{lemma}

\begin{proof}[Convergence]
The maximin utility, in each round, Pareto-dominates the previous one. This is because (i) all profiles that have, for any player, lower payoffs than their previous maximin utility, were eliminated and (ii) for any player, the maximin utility appears as a payoff somewhere in the game because the game is finite. Since the sequence $(\mathcal{S}_i(\Gamma))_{i\in\mathbb{N}}$ is decreasing and has its values in a finite set (in the powerset of all strategy profiles), it must converge and reach its limit. $\square$
\end{proof}

\subsection{Perfectly Transparent Equilibrium}

We can finally define the Perfectly Transparent Equilibrium for games in normal form, and characterize it with the previous sequence of sets of surviving strategy profiles. 

\begin{definition}[Perfectly Transparent Equilibrium for games in normal form]
Given a game in normal form $\Gamma$, with pure strategies and with no ties, a Perfectly Transparent Equilibrium is a strategy profile that never gets eliminated. It is immune to Necessary Rationality and Necessary Knowledge of Strategies, in the sense that the players willingly, rationally jointly play towards this outcome. The set of Perfectly Transparent Equilibria is $\mathcal{S}(\Gamma)$.
\end{definition}

The algorithm for computing the PTE for any game in normal form and with no ties follows: one iteratively eliminates, in each round, all strategy profiles that do not Pareto dominate the current tuple of maximin utilities.

\subsection{Properties of the Perfectly Transparent Equilibrium}
\label{section-results}

We now give two theoretical results concerning the PTE on games in normal form: uniqueness and Pareto optimality. The proofs are relatively straightforward.

\subsubsection{Uniqueness}

\begin{theorem}[Uniqueness]
Given a game in normal form, with pure strategies, and with no ties, if a PTE exists, then it is unique.
\end{theorem}

\begin{proof}[Uniqueness]
The sequence $(\mathcal{S}_i(\Gamma))_{i\in\mathbb{N}}$ is strictly decreasing until either the empty set or a singleton is reached. This is because there are no ties: for any player, the current maximin utility must, by definition of the max, be strictly greater than one of the payoffs for a different strategy (which exists for at least one player if the current set is not a singleton). The corresponding strategy profile will be eliminated in the next round.
As a consequence, the sequence can only converge towards a singleton (this is the unique PTE of the game) or the empty set (there are no PTEs for the game). $\square$
\end{proof}

\subsubsection{Pareto optimality}
\label{section-pareto-optimal}

\begin{theorem}[Pareto optimality]
Given a game in normal form, with pure strategies, and with no ties, if the PTE exists, then it is Pareto-optimal amongst all strategy profiles.
\end{theorem}

\begin{proof}[Pareto optimality]
If an existing PTE were not Pareto-optimal, then there would be a distinct strategy profile that Pareto-dominates the PTE. But this strategy profile would not have been eliminated in the last strictly decreasing round of elimination, because it would also Pareto-dominate the current tuple of maximin utilities and be a PTE as well. This contradicts uniqueness. $\square$
\end{proof}

\subsubsection{Existence}
\label{section-existence}

It is important to note that the equilibrium does not always exist. We provide counterexamples in Section \ref{section-counterexamples}.

One interesting property of games in normal form is the correlations of the payoffs. The two extremes of the spectrum are on the one hand zero-sum games, corresponding to strategic substitutes, and on the other hand strategic complements where agents have shared interests. As a measure of payoff correlation, we use the correlation coefficient of the payoff pairs within the game, i.e., for a game $\Gamma$\footnote{We used $\tau$ for strategy profiles to avoid a confusion of notation with the standard deviation, traditionally denoted with a $\sigma$.}:

$$c(\Gamma)=\frac{\frac{1}{|\Sigma|}\displaystyle\sum_{\tau\in\Sigma} (u_1(\tau)-\mu_1)\times (u_2(\tau)-\mu_2)}{\sigma_1\sigma_2}$$

with takes its values between -1 (strategic substitutes) and 1 (strategic complements), and where $\mu_i$ is the average payoff of agent $i$ and $\sigma_i$ is the standard deviation of the payoff of agent $i$ 

We empirically estimated the chances that the equilibrium exists for 2-player 3x3 games depending on the payoff correlation. Since we only consider ordinal payoffs, there are about $\frac{9!\times9!}{3!\times3!}$ such games, i.e., 3.66 billion. We took a simple random sample with replacement of 204,800,000 games obtained by randomly generating payoff permutations (of the integers 0 to 8) for each player. Thus, for these games, $\mu_1=\mu_2=4$ , and $\sigma_1=\sigma_2=2.581988897$. 

We calculated the payoff correlation of each game and rounded to the nearest multiple of 0.025, which gave us payoff correlation buckets. In each bucket, we computed the percentage of games for which (i) a PTE, (ii) a Nash Equilibrium (NE), and (iii) a unique NE exists\footnote{For a bit of theory: this is a sample mean estimator if we assign 1 to existence and 0 to non-existence. Testing for the existence of an equilibrium on a randomly picked game is a Bernoulli trial. Since the sample is with replacement, the estimator follows a binomial distribution. The size of the confidence intervals of the estimated existence ratio diminishes with the square root of the sample size and the estimation can be considered rather accurate when there are more than 1000 samples in a bucket.}.

The result is shown in Figure \ref{fig-existence-ratio}. Unsurprisingly (minimax theorem), the NE and unique NE existence ratios are close to each other for strategic substitutes, but strongly diverge for strategic complements; both the NE and the PTE follow a continuous pattern and tend towards always existing for strategic complements. We also observe that (i) in every bucket, the PTE exists (and is unique) more often than there is a unique NE. (ii) for a correlation below around -0.3 (i.e., for zero-sum games or games close to being zero-sum), the PTE exists noticeably more often than a (not necessarily unique) NE does, in addition to also not suffering from equilibrium multiplicity and to being Pareto-optimal. The PTE curve remains relatively flat across payoff correlations, while the NE existence ratio tends to decrease with strategic substitutes.

\begin{figure}
\centering{
\resizebox{0.7\textwidth}{!}{
\includegraphics{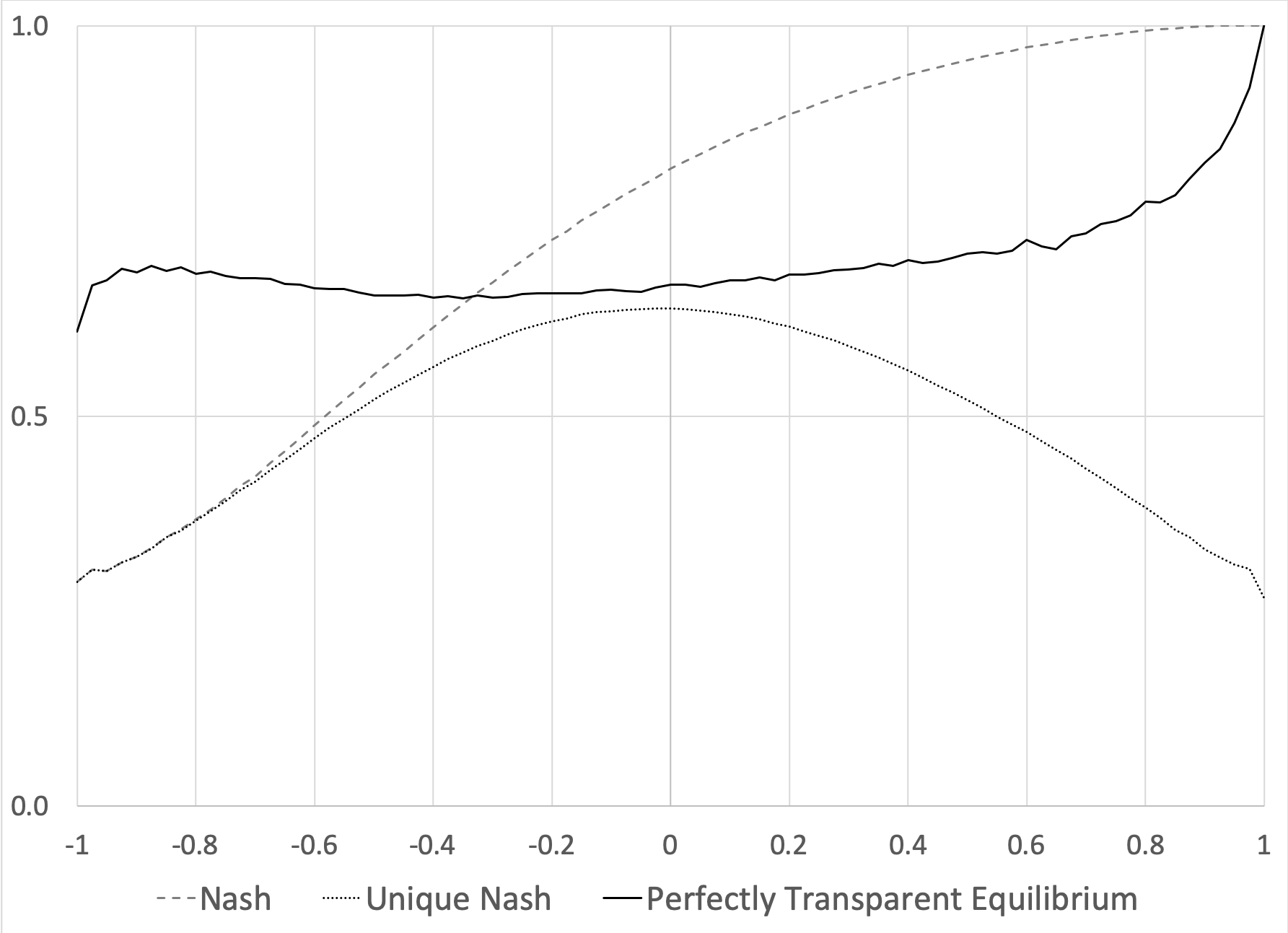}
}
}
\caption{Estimation of the existence ratio of the NE, unique NE and PTE depending on the payoff correlation of 3x3 two-player games. Precision is highest around the 0 bucket with several million data points and lowest at the extremes (-1 and 1 buckets) with a bit more than 500 data points, which is still reasonable -- and already more than 16,000 data points for payoff correlations in the -0.975 and 0.975 buckets.}
\label{fig-existence-ratio}
\end{figure}

\section{Examples}
\label{section-examples}

In this section, we give a few more examples of games and PTE computations, besides the prisoner's dilemma already solved in Section \ref{section-prisoner-dilemma-pte}. We start with an asymmetric game, then continue with a few games satisfying criteria that make them ``social dilemma'', both symmetric and asymmetric, and, in particular, show that the PTE solves these social dilemmas where Nash equilibria do not.

\subsection{Solving the PTE on an asymmetric game}

Figure \ref{fig-example2} shows an asymmetric game. Let us show that this game has a PTE.

In the first round, the players' maximin utilities are (1,1), with the maximin strategies being B guaranteeing 1 to the row player and D guaranteeing 1 to the column player. Strategy profiles AC (0,2) and BC (3,0) are eliminated because they do not Pareto-dominate the maximin tuple. Specifically, AC is preempted by B because the row player would, in any case, get a better payoff than 0 if he picked B instead. BC is preempted by D because the column player would, in any case, get a better payoff than 0 if she picked D instead.

In the second round, the new maximin utilities are (2,1). BD is preempted by A: if it were known to be the solution, the row player would pick A instead. Indeed, he knows that if he picked A instead, then the column player would have picked D as it was shown that AC is incompatible with the assumed row player's rationality, which applies in any possible world.

Only AD (2,3) remains, which is stable and immune to Necessary Rationality and Necessary Knowledge of Strategies. Knowing that AD (2,3) is the PTE, both players will stick to their strategy based on the above reasoning, as no other strategy profile is reasonable: If the row player had picked B instead, then the column player would have still picked D, and the payoff with AD would be less ($1<2$). If the column player had picked C instead, then either the row player would not have picked A because this would have been irrational, or BC would have been the known outcome and it is the column player that would not have been rational preferring 0 over a guarantee of 1.

So AD (2,3) is the outcome willingly and knowingly obtained by rational players under Necessary Rationality and Necessary Knowledge of Strategies: it is the PTE. It happens, in this case, to also be a Nash equilibrium because A and D are best responses to each other.

\begin{figure}
\begin{center}

\resizebox{\textwidth}{!}{
\begin{tabular}{lcccr}

\begin{tabular}{|r|c|c|}
\hline
& C & D \\
\hline
A & 0, 2& 2, 3\\
\hline
B & 3, 0& 1,1\\
\hline
\end{tabular}

&

$\Rightarrow$

&

\begin{tabular}{|r|c|c|}
\hline
& C & D \\
\hline
A & \cellcolor{black!25}0, 2& 2, 3\\
\hline
B & \cellcolor{black!25}3, 0& 1,1\\
\hline
\end{tabular}

&

$\Rightarrow$

&

\begin{tabular}{|r|c|c|}
\hline
& C & D \\
\hline
A & \cellcolor{black!25}0, 2& 2, 3\\
\hline
B & \cellcolor{black!25}3, 0&  \cellcolor{black!25}1,1\\
\hline
\end{tabular}

\end{tabular}
}
\end{center}
\caption{Iterated elimination of preempted strategy profiles in an asymmetric game. The tuples of maximin utilities are (1,1), then (2,3). (2,3) is the unique PTE and is immune against the common knowledge that it is the PTE.}
\label{fig-example2}
\end{figure}

\subsection{Public Goods Game}
\label{section-social-dilemma}

The Public Goods Game is one of the social dilemma presented by \citet{Capraro2015}.

\citet{Capraro2015} define a social dilemma as follows: (i) the game has exactly one Nash equilibrium and (ii) the game has exactly one improvement over the Nash equilibrium, i.e., exactly one profile, distinct from the Nash Equilibrium, that Pareto-dominates the Nash equilibrium. In some cases, the PTE is this exact outcome that Pareto-dominates the Nash equilibrium. In some other cases with several Pareto-optima, the PTE is a different Pareto-optimal outcome but does not Pareto-dominate the Nash equilibrium. 

Each player has one dollar and can decide to contribute part of it to the pool. Then, the pool is multiplied by a factor greater than one and smaller than the number of players and is then evenly distributed to the players.

For a player $i\in P$ contributing $x_i>0$, the worst that can happen is that nobody else contributed. In this case, he is left with less than a dollar. His maximin is thus 1 and is obtained by contributing nothing. At this first stage, we thus know that under Necessary Rationality and Necessary Knowledge of Strategies, all agents are guaranteed not to lose anything. We give an example with two players, a constant factor of 1.5 and contributions of 0, 0.5 and 1 in Figure \ref{fig-goods-game}. We can see that this game has no ties, that the PTE exists, is unique and coincides with the Hofstadter equilibrium, with everybody contributing as much as they can.

In the more general case with no ties, any number of players and any constant, the reasoning at every round of elimination is that a higher worst case than the minimum possible contribution (taking eliminated profiles into account) can be obtained by contributing a bit more (because there are no ties, and lower payoffs were eliminated in previous rounds). The elimination sequence then converges to the bottom-right, welfare-maximizing profile, which is the PTE. The PTE thus explains full cooperation on this game.

\begin{figure}
\begin{center}

\resizebox{\textwidth}{!}{
\begin{tabular}{lcccr}

\begin{tabular}{|r|c|c|c|}
\hline
& D & I & C \\
\hline
D & 1, 1& 1.375, 0.5 & 1.75, 0.75\\
\hline
I & 0.5, 1.375& 1.25,1.25& 1.625,1.125\\
\hline
C & 0.75, 1.7& 1.125,1.625& 1.5,1.5\\
\hline
\end{tabular}

&

$\Rightarrow$

&

\begin{tabular}{|r|c|c|c|}
\hline
& D & I & C \\
\hline
D & 1, 1&\cellcolor{black!25} 1.375, 0.5 & \cellcolor{black!25}1.75, 0.75\\
\hline
I & \cellcolor{black!25}0.5, 1.375& 1.25,1.25& 1.625,1.125\\
\hline
C & \cellcolor{black!25}0.75, 1.7& 1.125,1.625& 1.5,1.5\\
\hline
\end{tabular}

& 

$\Rightarrow$

\\
\\
\begin{tabular}{|r|c|c|c|}
\hline
& D & I & C \\
\hline
D &  \cellcolor{black!25} 1, 1&\cellcolor{black!25} 1.375, 0.5 & \cellcolor{black!25}1.75, 0.75\\
\hline
I & \cellcolor{black!25}0.5, 1.375&  1.25,1.25& \cellcolor{black!25} 1.625,1.125\\
\hline
C & \cellcolor{black!25}0.75, 1.7& \cellcolor{black!25} 1.125,1.625& 1.5,1.5\\
\hline
\end{tabular}

&

$\Rightarrow$

&

\begin{tabular}{|r|c|c|c|}
\hline
& D & I & C \\
\hline
D &  \cellcolor{black!25} 1, 1&\cellcolor{black!25} 1.375, 0.5 & \cellcolor{black!25}1.75, 0.75\\
\hline
I & \cellcolor{black!25}0.5, 1.375& \cellcolor{black!25}  1.25,1.25& \cellcolor{black!25} 1.625,1.125\\
\hline
C & \cellcolor{black!25}0.75, 1.7& \cellcolor{black!25} 1.125,1.625& 1.5,1.5\\
\hline
\end{tabular}

\end{tabular}
}
\end{center}
\caption{The Goods game with two players, a constant factor of 1.5 and discrete possible contributions: nothing (D), half (I), everything (C). The PTE and the Hofstadter equilibrium coincide on (1.5, 1.5). The Nash equilibrium is, however, (1,1) because C is a dominant strategy for all players.}
\label{fig-goods-game}
\end{figure}

\subsection{Bertrand Competition}

Another example of a game by \citet{Capraro2015} is the Bertrand Competition game, where companies set their prices. Only those with the lowest price sell their product, sharing among ties. Figure \ref{fig-bertrand-game} shows how to solve this game for its PTE, which exists in this case.

The Perfectly Transparent Equilibrium coincides with all companies picking the highest price, which is the welfare-maximizing equilibrium. Indeed, a single round of elimination eliminates all outcomes not on the diagonal, because the maximum worst gain when picking the lowest price is the lowest price divided by the number of players. The worst gain on all other strategies is zero. Since all strategies outside the diagonal involve 0 for some player, they are all eliminated. A second round of elimination only retains the welfare-maximizing equilibrium. The PTE thus also explains full cooperation on this game.

\begin{figure}
\begin{center}

\resizebox{\textwidth}{!}{
\begin{tabular}{lccr}

\begin{tabular}{|r|c|c|c|}
\hline
& 2 & 4 & 6 \\
\hline
2 & 1, 1& 2, 0 & 2, 0\\
\hline
4 &0, 2& 2, 2&  4, 0\\
\hline
6 & 0, 2& 0, 4& 3, 3\\
\hline
\end{tabular}

&

$\Rightarrow$

&

\begin{tabular}{|r|c|c|c|}
\hline
& 2 & 4 & 6 \\
\hline
2 & 1, 1& \cellcolor{black!25}2, 0 & \cellcolor{black!25}2, 0\\
\hline
4 &\cellcolor{black!25}0, 2& 2, 2&  \cellcolor{black!25}4, 0\\
\hline
6 & \cellcolor{black!25}0, 2& \cellcolor{black!25}0, 4& 3, 3\\
\hline
\end{tabular}

& 

$\Rightarrow$

\begin{tabular}{|r|c|c|c|}
\hline
& 2 & 4 & 6 \\
\hline
2 & \cellcolor{black!25}1, 1& \cellcolor{black!25}2, 0 & \cellcolor{black!25}2, 0\\
\hline
4 &\cellcolor{black!25}0, 2& \cellcolor{black!25}2, 2&  \cellcolor{black!25}4, 0\\
\hline
6 & \cellcolor{black!25}0, 2& \cellcolor{black!25}0, 4& 3, 3\\
\hline
\end{tabular}

\end{tabular}
}
\end{center}
\caption{The Bertrand Competition game with two players with prices settable to 2, 4 or 6. The PTE and the Hofstadter equilibrium coincide on the welfare-maximizing equilibrium. The Nash equilibrium is, however, (1,1) because 2 is a dominant strategy for all players.}
\label{fig-bertrand-game}
\end{figure}

\subsection{Traveler's dilemma}

Another example of a game by \citet{Capraro2015} is the Traveler's dilemma game. Two travelers who lost their luggage can ask for compensation. If they both ask for the same amount, they get it. Otherwise, the one who asked for the lowest amount gets it with a bonus, and the other gets that same lowest amount with the bonus deducted as a penalty. Figure \ref{fig-traveler-game} shows an instance of its game solved for its (existing) PTE.

The Perfectly Transparent Equilibrium coincides with both players picking the highest compensation number. Indeed, the maximum worst case is the lowest price, because all other strategies let to a worst-case scenario of the lowest price minus the penalty. In the second round, the maximum worst case becomes the next lowest price, and so on, until the maximum price is the only one left for both players. The PTE thus also explains full cooperation on this game.

\begin{figure}
\begin{center}

\resizebox{\textwidth}{!}{
\begin{tabular}{lcccr}

\begin{tabular}{|r|c|c|c|}
\hline
& 2 & 3 & 4 \\
\hline
2 & 2, 2& 3, 1 & 3, 1\\
\hline
3 &1, 3& 3, 3&  4, 2\\
\hline
3 & 1, 3& 2, 4& 4, 4\\
\hline
\end{tabular}

&

$\Rightarrow$

&

\begin{tabular}{|r|c|c|c|}
\hline
& 2 & 3 & 4 \\
\hline
2 & 2, 2& \cellcolor{black!25}3, 1 & \cellcolor{black!25}3, 1\\
\hline
3 &\cellcolor{black!25}1, 3& 3, 3&  4, 2\\
\hline
3 & \cellcolor{black!25}1, 3& 2, 4& 4, 4\\
\hline
\end{tabular}

& 

$\Rightarrow$

\\
\\

\begin{tabular}{|r|c|c|c|}
\hline
& 2 & 3 & 4 \\
\hline
2 & \cellcolor{black!25}2, 2& \cellcolor{black!25}3, 1 & \cellcolor{black!25}3, 1\\
\hline
3 &\cellcolor{black!25}1, 3& 3, 3&  \cellcolor{black!25}4, 2\\
\hline
3 & \cellcolor{black!25}1, 3& \cellcolor{black!25}2, 4& 4, 4\\
\hline
\end{tabular}
&

$\Rightarrow$
& 

\begin{tabular}{|r|c|c|c|}
\hline
& 2 & 3 & 4 \\
\hline
2 & \cellcolor{black!25}2, 2& \cellcolor{black!25}3, 1 & \cellcolor{black!25}3, 1\\
\hline
3 &\cellcolor{black!25}1, 3& \cellcolor{black!25}3, 3&  \cellcolor{black!25}4, 2\\
\hline
3 & \cellcolor{black!25}1, 3& \cellcolor{black!25}2, 4& 4, 4\\
\hline
\end{tabular}

\end{tabular}
}
\end{center}
\caption{The Traveler's dilemma game with two players with prices settable to 2, 3 or 4and a bonus/penalty of 2. The PTE coincides with the welfare-maximizing equilibrium. The Nash equilibrium is, however, (2,2).}
\label{fig-traveler-game}
\end{figure}

\subsection{An asymmetric social dilemma}

All social dilemmas described by \citet{Capraro2015} are symmetric games. The PTE solves them by providing counterfactual structures based on Necessary Rationality and Necessary Knowledge of Strategies that lead to a unique outcome: the welfare-maximizing profile. The Hofstadter equilibrium also explains this same profile with the symmetric diagonal argument.

But the definition of a social dilemma given by \citet{Capraro2015} does not require in general symmetry. Figure \ref{figure-asymmetric-social-dilemma} shows an asymmetric social dilemma, with the iterative deletion procedure of the PTE made explicit. This game was found with a filtering query on a sampled dataset of games by Felipe \cite{Sulser2019} during his Master's thesis. We can see that the PTE solves this asymmetric social dilemma.

Note, however, that it is not always the case that the PTE Pareto-improves a Nash equilibrium. There are thus social dilemmas in which the PTE is Pareto-optimal, but does not Pareto-dominate the Nash equilibrium.

\begin{figure}
\begin{center}

\resizebox{\textwidth}{!}{
\begin{tabular}{lcccr}

\begin{tabular}{|r|c|c|c|}
\hline
& A & B & C \\
\hline
D & 6, 8& 1, 2 & 4, 4\\
\hline
E &3, 1& 0, 7&  2, 3\\
\hline
F & 7, 0& 8, 5& 5, 6\\
\hline
\end{tabular}

&

$\Rightarrow$

&

\begin{tabular}{|r|c|c|c|}
\hline
& A & B & C \\
\hline
D & 6, 8&\cellcolor{black!25} 1, 2 &\cellcolor{black!25} 4, 4\\
\hline
E &\cellcolor{black!25}3, 1& \cellcolor{black!25}0, 7&  \cellcolor{black!25}2, 3\\
\hline
F & \cellcolor{black!25}7, 0& 8, 5& 5, 6\\
\hline

\end{tabular}

& 

$\Rightarrow$

&

\begin{tabular}{|r|c|c|c|}
\hline
& A & B & C \\
\hline
D & 6, 8&\cellcolor{black!25} 1, 2 &\cellcolor{black!25} 4, 4\\
\hline
E &\cellcolor{black!25}3, 1& \cellcolor{black!25}0, 7&  \cellcolor{black!25}2, 3\\
\hline
F & \cellcolor{black!25}7, 0& \cellcolor{black!25}8, 5& \cellcolor{black!25}5, 6\\
\hline
\end{tabular}

\end{tabular}
}
\end{center}
\caption{An asymmetric social dilemma: the PTE coincides with the welfare-maximizing profile. The Nash equilibrium, (5, 6), is Pareto-dominated by the PTE.}
\label{figure-asymmetric-social-dilemma}
\end{figure}

\section{Literature and related work on non-Nashian game theory}
\label{section-literature}

There is a growing literature of solution concepts in non-Nashian Game Theory. They all have in common the change in assumption: while in the Nash paradigm, the opponent's strategies are held fixed while optimizing one's utility, this line of research assumes that there may be a counterfactual dependence between the agent's decisions.

It is important to understand that the non-Nashian line of reasoning is not Evidential Decision Theory as explained in Section \ref{section-newcomb}. Indeed, the reasoning involves potentially asymmetric counterfactual dependencies, expressed in English as subjunctive conditionals, while EDT uses conditional probabilities.

As a rule of thumb, there are two main categories of non-Nashian approaches. The first approach is found for example in \citet{Halpern:2013aa}, who coined it as ``translucent''. While they drop the Nash assumption that opponent strategies are counterfactually independent of a player's choice of strategy, they assume a weaker version of rationality than Necessary Rationality. They call this weaker version Common Counterfactual Belief of Rationality, as we will see shortly. The other approach, which is the one taken in this paper, is fully transparent. It not only drops the Nash assumption that opponent strategies are counterfactually independent of a player's choice of strategy but replaces it with different assumptions regarding the counterfactuals: namely, that even if a player had picked a different strategy, both players would still have been rational and would still have known the outcome of the game.

Translucency has, overall, a tendency to Pareto-improve Nash equilibria (for strong solution concepts such as Shiffrin's), or to subsume Nash equilibria\footnote{This is only a conjecture formulated in vague wording at this point, i.e., that non-fully-transparent translucency either subsumes or Pareto-improves Nashian results, and full transparency behaves like a singularity where Nashian results no longer apply, and the outcomes suddenly diverge from non-fully-transparent translucency. From a modal logic perspective, the difference lies in the presence of impossible possible worlds \citep{Rantala1982}  due to the strong assumption of epistemic omniscience in all possible worlds, used to model full transparency.} (for weaker concepts such as Halpern's and Pass's minimax-Rationalizability). The fully transparent approach is more orthogonal and sometimes coincides with a Nash equilibrium, sometimes Pareto-improves a Nash equilibrium, sometimes has no Pareto relationship with the Nash Equilibrium (some agents get more, some less).

\subsection{Common Counterfactual Belief of Rationality and Minimax Rationalizability (Non-Nashian, translucent)}

A classical assumption made in game theory is Common Belief (or Knowledge)\footnote{In epistemic logic, the difference is that a player can believe that an event holds even it does not. Knowledge is thus a stronger assumption than belief.} of Rationality: all players are rational, and believe (or know) that they all are, and believe that they all believe that they all are, and so on. In other words, common knowledge of rationality is assumed in the actual world.

In non-Nashian Game theory, rationality is also \emph{counterfactually} assumed. In this paper, we assume the strong criterion of Necessary Rationality: in any possible world, all players are rational. Necessary Rationality implies Common Knowledge of Rationality, but the converse is not true: Common Knowledge of Rationality does not imply that rationality also holds counterfactually, that is, that rationality would still hold if the agents made different decisions.

Joseph Halpern and Rafael Pass explicitly name---and formally define---the weaker assumption of Common Counterfactual Belief of Rationality (CCBR). Firstly, CCBR means that all players are rational. Secondly, it means that each one of them counterfactually believes that everyone \emph{else} is rational, which means that they believe that, even if an agent had acted differently, all \emph{other} agents would have been rational. This is where it is weaker: Necessary Rationality more strongly implies that, if an agent had acted differently, all agents \emph{including him/her} would have been rational. Finally CCBR also recursively means that all agents counterfactually believe that everyone \emph{else} counterfactually believes that everyone \emph{else} is rational, and so on.

As is done in this paper, CCBR does not assume that the opponents' strategies would be unchanged if a player unilaterally changed his strategy. Dropping this assumption leaves room for alternative assumptions, in their case CCBR. Halpern and Pass make no particular assumptions regarding Perfect Prediction, which we do in this paper.

Halpern and Pass define minimax rationalizability to select strategies that make sense under CCBR. The algorithmic  characterization of minimax rationalizability is by iterated deletion of strategies that are minimax dominated. Informally, in a game in normal form, a strategy of player $i$ is minimax dominated if there exists another strategy that guarantees him a higher payoff no matter what the opposite player does or would do. In other words, there is another strategy for which the minimum payoff is greater than the maximum payoff of the dominated strategy.

The intuition behind this definition is that, even if the player considers that the opponent's strategy is counterfactually dependent on his choice, a minimax-dominated strategy will never be a good choice as the payoffs will nevertheless always be less, no matter what the assumed counterfactual dependence is. If a player P picked a minimax-dominated strategy, even if the opponent's strategy is the best possible case for P, there is another strategy that would give him a higher payoff even if the (then possibly different) opponent's strategy were the worst possible case. This is to be put in contrast with ``classical'' rationalizability, in which	strategies that are not best responses are eliminated, which is a weaker requirement for elimination. In other words, a strategy that is minimax-dominated would also be eliminated according to classical rationalizability.

Figure \ref{fig-minimax-dominated} shows an example of a game in which some strategies are minimax-dominated and can thus be eliminated under CCBR. It can be seen in these games that fewer outcomes are eliminated than with classical rationalizability. Indeed, with A and D eliminated, E, which is minimax-rationalizable, is never a best response, which makes it non-rationalizable. If we drop unilateral deviations, then CE may still be reasonable if the column player believes that, if she had picked F (and this leaked), the row player would have picked B: $8>7$. This is something that cannot be captured with unilateral deviations in Causal Decision Theory.

\begin{figure}
\begin{center}

\resizebox{\textwidth}{!}{
\begin{tabular}{lcccr}

\begin{tabular}{|r|c|c|c|}
\hline
& D & E & F\\
\hline
A & 1,5& 2,3 & 4,1\\
\hline
B & 3,2 & 6,6 & 9,7\\
\hline
C & 5,4 & 7,8 & 8,9\\
\hline
\end{tabular}

&

$\Rightarrow$

&

\begin{tabular}{|r|c|c|c|}
\hline
& D & E & F\\
\hline
\cellcolor{black}\textcolor{white}A & \cellcolor{black}\textcolor{white}{1,5}& \cellcolor{black}\textcolor{white}{2,3} & \cellcolor{black}\textcolor{white}{4,1}\\
\hline
B & 3,2 & 6,6 & 9,7\\
\hline
C & 5,4 & 7,8 & 8,9\\
\hline
\end{tabular}

&

$\Rightarrow$

&

\begin{tabular}{|r|c|c|c|}
\hline
& \cellcolor{black}\textcolor{white}D & E & F\\
\hline
\cellcolor{black}\textcolor{white}A & \cellcolor{black}\textcolor{white}{1,5}&\cellcolor{black}\textcolor{white}{2,3} & \cellcolor{black}\textcolor{white}{4,1}\\
\hline
B & \cellcolor{black}\textcolor{white}{3,2} & \cellcolor{black!25}6,6 & 9,7\\
\hline
C & \cellcolor{black}\textcolor{white}{5,4} & \cellcolor{black!25}7,8 & \cellcolor{black!25}8,9\\
\hline
\end{tabular}

\end{tabular}
}
\end{center}
\caption{A 3x3 game in normal form, for which strategy A is minimax-dominated by both C, and (after eliminating C) D is minimax-dominated by either E or F. The iteration goes from left to right, and eliminated profiles are marked in black. A and D are not minimax-rationalizable, in other words, are not rational according to CCBR, even if each player considers the strategy of the other player not to be fixed. Strategy profiles marked in black are also not rationalizable. Those marked in light gray are minimax-rationalizable, but not rationalizable.}
\label{fig-minimax-dominated}
\end{figure}

\begin{definition}[minimax rationalizability]: Given a game $(P, \Sigma, u)$, given a player $i$, a strategy $\sigma_i$ is minimax-dominated\footnote{In the original paper \citep{Halpern:2013aa}, the opponent's strategy is taken from a subset of the opponent's strategies to account for successive eliminations. We are leaving out this aspect here for pedagogical reasons. Indeed, one can also mentally update $\Sigma$ in place as strategies get eliminated.} if

$$\exists \upsilon_i \in \Sigma_i, \min_{\tau_{-i}\in\Sigma_{-i}} u_i(\upsilon_i, \tau_{-i}) > \max_{\tau_{-i}\in\Sigma_{-i}} u_i(\sigma_i, \tau_{-i})$$

\end{definition}

Joseph Halpern and Rafael Pass give alternative characterizations of minimax-rationalizabi\-lity, but iterated deletion is the most intuitive one, and the one that we will use for our proof. The concept of minimax-rationalizability, defined on strategies, extends to strategy profiles, i.e., a strategy profile can be considered to be minimax-rationalizable if all the strategies it is made of are all minimax-rationalizable.

A strong feature of minimax-rationalizability is that the result is independent of the order in which strategies are eliminated. The cost to pay is that it is a weak criterion in the sense that for many games, most strategies are minimax-rationalizable. For example, the games shown in Figures \ref{fig-prisoner-dilemma}, \ref{fig-chicken-game} and \ref{fig-coordination-game} have no minimax-dominated strategies. The PTE is a strong criterion because it is at most unique. However, it requires that the outcomes are eliminated in successive steps, which means that there is an underlying structure in the (counterfactual) levels of logical omniscience as the reasoning progresses. This is formalized in details in \cite{Fourny2018b} where it is conjectured in particular that necessary logical omniscience, necessary rationality and necessary knowledge of strategies form an impossibility triangle. In the framework of perfect prediction, the logical reasoning needs to be bootstrapped starting with the entire set of outcomes, which is formally achieved with non-normal-worlds \citep{Rantala1982} in which anything is possible and nothing is necessary.

\begin{figure}
\begin{center}

\resizebox{\textwidth}{!}{
\begin{tabular}{lcccr}

\begin{tabular}{|r|c|c|c|}
\hline
& D & E & F \\
\hline
A & 1, 1& 2, 2 & 3, 4\\
\hline
B & 4, 5 & 6, 8 & 7, 9\\
\hline
C & 5, 6 & 8, 3 & 9, 7\\
\hline
\end{tabular}

&

$\Rightarrow$

&

\begin{tabular}{|r|c|c|c|}
\hline
& D & E & F \\
\hline
A & \cellcolor{black!25}1, 1& \cellcolor{black!25}2, 2 & \cellcolor{black!25}3, 4\\
\hline
B & \cellcolor{black!25}4, 5 & 6, 8 & 7, 9\\
\hline
C & 5, 6 & \cellcolor{black!25}8, 3 & 9, 7\\
\hline
\end{tabular}

&

$\Rightarrow$

\\

\\

$\Rightarrow$

\begin{tabular}{|r|c|c|c|}
\hline
& D & E & F \\
\hline
A & \cellcolor{black!25}1, 1& \cellcolor{black!25}2, 2 & \cellcolor{black!25}3, 4\\
\hline
B & \cellcolor{black!25}4, 5 & 6, 8 & 7, 9\\
\hline
C & \cellcolor{black!25}5, 6 & \cellcolor{black!25}8, 3 & \cellcolor{black!25}9, 7\\
\hline
\end{tabular}

&

$\Rightarrow$

&

\begin{tabular}{|r|c|c|c|}
\hline
& D & E & F \\
\hline
A & \cellcolor{black!25}1, 1& \cellcolor{black!25}2, 2 & \cellcolor{black!25}3, 4\\
\hline
B & \cellcolor{black!25}4, 5 & \cellcolor{black!25}6, 8 & 7, 9\\
\hline
C & \cellcolor{black!25}5, 6 & \cellcolor{black!25}8, 3 & \cellcolor{black!25}9, 7\\
\hline
\end{tabular}

\end{tabular}
}
\end{center}
\caption{An example of a game in which the PTE is not minimax-rationalizable. We show the iterated elimination of profiles towards the PTE.}
\label{fig-ccbr-counterexample}
\end{figure}

\begin{figure}
\begin{center}

\resizebox{\textwidth}{!}{
\begin{tabular}{lcccr}

\begin{tabular}{|r|c|c|c|}
\hline
& D & E & F \\
\hline
A & 1, 1& 2, 2 & 3, 4\\
\hline
B & 4, 5 & 6, 8 & 7, 9\\
\hline
C & 5, 6 & 8, 3 & 9, 7\\
\hline
\end{tabular}

&

$\Rightarrow$

&

\begin{tabular}{|r|c|c|c|}
\hline
& D & E & F \\
\hline
\cellcolor{black!25}A & \cellcolor{black!25}1, 1& \cellcolor{black!25}2, 2 & \cellcolor{black!25}3, 4\\
\hline
B & 4, 5 & 6, 8 & 7, 9\\
\hline
C & 5, 6 &8, 3 & 9, 7\\
\hline
\end{tabular}

&

$\Rightarrow$

\\

\\

$\Rightarrow$

\begin{tabular}{|r|c|c|c|}
\hline
& \cellcolor{black!25}D & E & F \\
\hline
\cellcolor{black!25}A & \cellcolor{black!25}1, 1& \cellcolor{black!25}2, 2 & \cellcolor{black!25}3, 4\\
\hline
B & \cellcolor{black!25}4, 5 & 6, 8 & 7, 9\\
\hline
C & \cellcolor{black!25}5, 6 & 8, 3 & 9, 7\\
\hline
\end{tabular}

&

$\Rightarrow$

&

\begin{tabular}{|r|c|c|c|}
\hline
&  \cellcolor{black!25}D & E & F \\
\hline
 \cellcolor{black!25}A & \cellcolor{black!25}1, 1& \cellcolor{black!25}2, 2 & \cellcolor{black!25}3, 4\\
\hline
 \cellcolor{black!25}B & \cellcolor{black!25}4, 5 & \cellcolor{black!25}6, 8 &  \cellcolor{black!25}7, 9\\
\hline
C & \cellcolor{black!25}5, 6 & 8, 3 & 9, 7\\
\hline
\end{tabular}

\\

\\

$\Rightarrow$

\begin{tabular}{|r|c|c|c|}
\hline
& \cellcolor{black!25}D & \cellcolor{black!25}E & F \\
\hline
\cellcolor{black!25}A & \cellcolor{black!25}1, 1& \cellcolor{black!25}2, 2 & \cellcolor{black!25}3, 4\\
\hline
\cellcolor{black!25}B & \cellcolor{black!25}4, 5 & \cellcolor{black!25}6, 8 & \cellcolor{black!25}7, 9\\
\hline
C & \cellcolor{black!25}5, 6 & \cellcolor{black!25}8, 3 & 9, 7\\
\hline
\end{tabular}

&

&

\end{tabular}
}
\end{center}
\caption{An example of a game in which the PTE is not minimax-rationalizable. We show the iterated elimination of strategies that are not minimax-rationalizable.}
\label{fig-ccbr-counterexample-2}
\end{figure}

A counterexample was obtained by running the computations on a large number of games. No counterexample was found on $2\times2$ games. Among the 204+ million $3\times3$ games randomly tried out\footnote{This is a sample of the entire space of $3\times3$ games.} in Section \ref{section-existence}, about 137 million had a PTE. Of another random sample of about 12 million games, a bit more than 8000 was not minimax-rationalizable, which is less than 0.1\%. This counterexample is shown in Figures \ref{fig-ccbr-counterexample} and \ref{fig-ccbr-counterexample-2}, with both solution concepts are shown. Part of the reason is that minimax rationalizability eliminates strategies (rows and columns), while necessary rationality and Necessary Knowledge of Strategies eliminate strategy profiles (single outcomes).

Minimax-rationalizability is based on iterated elimination of strategies, which is also indirectly iterated elimination of strategy profiles, but in entire row or column batches. While it is true that all strategy profiles that get eliminated in the first round in minimax-rationalizability are also eliminated in the first round in the PTE, the converse is not true and further rounds may start to deviate, converging to different final sets. Figure \ref{fig-ccbr-counterexample-2} gives the rounds of elimination of strategies that are not minimax-rationalizable for the same game. We can see that only C and F are minimax-rationalizable, leaving only CF under CCBR

A closer look shows that the order in which we eliminate strategy profiles in the successive rounds matters: a strategy profile that the PTE eliminates in the first round cannot be used for further preemption. It is thus important in the PTE that, at each round, all strategy profiles that are known to be inconsistent with Necessary Rationality and Necessary Knowledge of Strategies  must be eliminated, using all the knowledge on impossible profiles available from the previous round: otherwise, a non-eliminated strategy profile may mistakenly affect (upwards) the maximin utility, eliminating strategy profiles that should not be eliminated.

This shows that, under CCBR, which is the translucent case, there is some opacity: in the third round of elimination, in Figure \ref{fig-ccbr-counterexample-2}, B is eliminated because 6 and 7 are both below the surviving 8 and 9. However, in the case of Necessary Rationality and Necessary Knowledge of Strategies, the players commonly know in all possible worlds that CE, which is not individually rational, is not possible. Likewise, they commonly know in all possible worlds that, since AE is not individually rational either, then in any possible world in which the column player picks E, the row player picks B. With this knowledge, obtained after two rounds of reasoning under Necessary Rationality and Necessary Knowledge of Strategies, CF is eliminated as well, as the column player would deviate to E to get 8 instead of 7. Strategy C thus is never picked in any possible world, and cannot be used to argue that strategy B is not possible.

This means that there is some singularity between a translucent setting and a fully transparent setting in which extra knowledge can be used to eliminate single profiles rather than strategies as a whole. This singularity does not exist for symmetric games, as we will see in Section \ref{section-symmetric-games}. Indeed, on symmetric games, any PTE is always minimax-rationalizable. Halpern and Pass also include a discussion of individual rationality of strategy profiles, as also introduced in Section \ref{section-individual-rationality}, which shows that it is relevant not only in the Nashian paradigm as known from the Folk theorems but also to non-Nashian reasoning.

Minimax-rationalizability and individual rationality, for strategy profiles, are not subsuming each other in any way: an individually rational strategy profile may not survive iterated minimax-deletion, and not all strategy profiles that survive it are individually rational. As Halpern and Pass point out, an individually rational strategy profile will always survive the first round of minimax elimination but may get eliminated in the second. The intuitive reason is that, after a round of elimination, this strategy profile may ``lose'' its individual rationality because the elimination of some strategies can increase the threshold required for individual rationality.

\subsection{Perfect Prediction Equilibrium  (Non-Nashian, transparent)}

The counterpart of the PTE on extensive form with perfect information was published by \citet{Fourny2018}, with the same assumptions of Necessary Knowledge of Strategies and Necessary Rationality (even though the terminology in this older paper may be different). It is also known as the Projected Equilibrium, which was the initial name used by Jean-Pierre Dupuy. The PPE is the natural counterpart of the PTE for games in extensive form. It has many features in common. First, it is based on an iterated elimination of preempted outcomes. Preemption is also done using the minimum guaranteed payoff by the strategy or move used to deviate. Like the PTE, the PPE is unique and Pareto-optimal.

However, the PPE differs from the PTE in that it always exists, and must solve Grandfather's paradoxes. In particular, given a game in extensive form and its PPE, if we convert the game to a normal form as is done in Nashian game theory, the PTE of this converted normal form will not always match the PPE. This is because the normal form does not carry the causal dependencies of the extensive forms. While these causal dependencies are not needed in the Nash equilibrium (because the past is counterfactually independent of the future), they are paramount in Perfect Prediction settings: the consistency of the game timeline must be preserved and the successive choices of the players must cause\footnote{in the sense that the former do not preempt the latter. In relativistic terms, we think of causality in the sense of inclusion in the future light cone.} each other: a decision cannot be made at a node $n$ if the player playing at the parent node did not pick $n$.

\subsection{Superrationality (Non-Nashian, transparent)}

\citet{Hofstadter1983} defined superrational equilibria on symmetric games in normal form. Superrationality is discussed in Section \ref{section-superrationality} on symmetric games.

\subsection{Shiffrin's Joint-Selfish-Rational equilibrium (Non-Nashian, translucent)}

\citet{Shiffrin2009} suggest an alternative approach for the discovery of Pareto optima in extensive form games, the Joint-Selfish-Rational equilibrium (JSRE). The approach differs from the PPE \citep{Fourny2018} and PTE (this paper), in that outcomes that have been eliminated can still be considered as deviations in subsequent rounds. The JSRE reasoning starts with the Subgame-Perfect Equilibrium, and navigates up and down the tree finding Pareto optimizations of successive interim equilibria. This leads to the same solution as the PPE in many games, but diverges from the PPE on other games. The JSRE in its original form is explained on several examples and has been formalized in pseudo-code by Felipe \cite{Sulser2019} in the general case. It has an exponential complexity to calculate as the number of players increases. The JSRE is currently undergoing a redesign by Rich Shiffrin to account for use cases in which the predicted equilibrium was not fully satisfactory in practice.

We believe that the divergence between the PPE and the JSRE paradigm is due to a fundamental axiomatic disagreement, in that the PPE reasoning (like the PTE reasoning) is based on the absence of contingencies, due to the uniqueness of the equilibrium, and on reasonings only on the equilibrium path. The JSRE uses counterfactual implications outside the equilibrium path as the SPE does, which we believe make it prone to the Backward Induction Paradox.

\subsection{Program Equilibrium (Non-Nashian, translucent)}

\cite{Tennenholtz2004} introduces the concept of an equilibrium in normal form, the Program Equilibrium, in which the players provide computer programs instead of strategies. The programs can read each other's source code, which means that the computations are transparent to each other. This is increasingly relevant in the context of smart contracts and, for example, the Ethereum blockchain, where smart contracts can look at each other. They show that any individual outcome can be obtained with this setup.

A crucial difference with the PTE is that, in the Program Equilibrium, the programs can read each other and deviate to a punishment in case they are not identical, however, there is no concept to enforce that the programs \emph{would also have been the same} if they had picked a different strategy. In the PTE, there is Necessary Rationality and Necessary Knowledge of Strategies, which means that the programs should be transparent to each other not only in the actual world but in all possible worlds. For any game, the Program Equilibria are the individually rational outcomes, which are also known to be obtained as steady states of repeated games (Folk theorem).

It remains an open avenue of research how the Program Equilibrium can be adapted to give a setup that exactly matches the transparency assumptions described in this paper.

\subsection{Second-Order Nash Equilibrium (Non-Nashian, translucent)}

\cite{bilo2011} introduces the concept of Second-Order Nash Equilibrium. Like the PTE, they are interested in one-shot games as opposed to repeated games. They extend the set of equilibria to a superset of Nash equilibria, making it a bigger set of candidate outcomes and a weaker condition, as are individual rationality, rationalizability or minimax-rationalizability.

Put simply, Second-Order Nash Equilibria differ from the PTE in that deviations are made sequentially in the actual world, as opposed to counterfactually as hypothetical alternatives. Indeed, while Second-Order Nash Equilibria model the consequences of a deviation of strategy in terms of several improving steps leading eventually to a Nash equilibrium that may or may not improve an agent's payoff, the PTE formally models deviations of strategies into a recursive structure of Lewisian closest-state functions, as is also done for minimax-rationalizability, but with the players being rational and making correct predictions in all states. 

The PTE, as opposed to Second-Order Nash Equilibra, shows that in a fully transparent setting, the laws of logic dictate that at most one outcome can be reached. It is thus a stronger condition. There is no known relation of inclusion between the two concepts as of today, and we suspect that there might be counter-examples.

\subsection{Stalnaker-Bonanno equilibrium (Nashian)}

The PTE distinguishes itself from other solution concepts such as minimax-rationalizability, in that it iteratively eliminates single strategy profiles rather than entire strategies. As it turns out, there are other examples of solution concepts that eliminate individual profiles, including in Nashian literature. A prominent example was originally formulated by \citet{Stalnaker1994} in the presence of Common Belief of Rationality as well as the assumption that what is believed is actually true (but not recursively). This leads to a slightly stronger definition of rationality, in which it is considered irrational, given a specific opponent's decision $\sigma_{-i}$, to play a strategy $\sigma_i$ that is weakly dominated by another strategy (smaller-or-equal payoffs), and such that for that one specific opponent's decision, the inequality is strict. This eliminates specifically the non-optimal profile ($\sigma_{i}, \sigma_{-i})$, but not the entire strategy $\sigma_i$. The process is iterated until it converges to a set of remaining profiles. \citet{Bonanno2008} gave a syntactic characterization of this solution concept to complement Stalnaker's work.

This equilibrium concept differs from the PTE, because payoff comparisons are done by fixing the opponent's strategy (weak domination of strategies), which is the Nashian free choice assumption where decisions are taken independently of each other.

\subsection{Perfect Cooperative Equilibrium (Non-Nashian, translucent)}

The Perfect Cooperative Equilibrium (PCE) was introduced by \citet{Rong2014} to address social dilemmas. It is defined on games in normal form for any number of players. It is based on a ``maximax'' approach in that one computes the best possible payoff, under the constraint that the opponents are best-responding. All strategy profiles that Pareto-dominate the obtained payoffs are PCE. A PCE always Pareto-dominates all Nash equilibria.

The Perfect Cooperative Equilibrium is designed in the context of repeated games in which players ``best respond to what they have learned'', i.e., the idea is that the two agents can converge to a Perfect Cooperative Equilibrium rather than to a Nash Equilibrium.

\citet{Rong2014} also introduce variants such as the M-PCE, where the threshold above which a PCE is obtained is offset by the same, maximum possible amount -- possible negative -- for all players until exactly one remains. Cooperative Equilibria (CE) are also defined for two-player games with a slightly weaker condition than the PCE.

It is worth noting that, in the case of a zero-sum game with two players, the PCE coincides with individual rationality, because maximizing the opponent's payoff (anticipating their best response) is identical to minimizing one's own payoff (considering the worst-case scenario).

\subsection{Translucent equilibrium (Non-Nashian, translucent)}

The translucent equilibrium was introduced in \citet{Capraro2015} as a weak solution concept that captures translucency. Translucent equilibria are algorithmically obtained, for pure strategies\footnote{Capraro and Halpern actually define it also on mixed strategies.}, as those that Pareto-dominate the tuple of second-lowest ``minimin'', in a single round of elimination. In other words, an agent considers the worst-case scenario for each one of his strategies. She then looks at the obtained payoffs and looks at the second-lowest. Any strategy profiles that yield a payoff below this threshold to this agent are eliminated.

This is thus similar to individual rationality, but with the second-lowest minimum rather than the maximum. For this reason, in a pure strategy setting, an individually rational strategy profile is also a translucent equilibrium because the threshold is stricter. A Nash Equilibrium is always translucent, and the PTE is also always translucent, because it is always individually rational. The Translucent equilibrium is thus the lowest known common denominator between the PTE and the Nash equilibrium.

\subsection{Correlated equilibrium (Nashian)}

A common question asked about the Perfectly Transparent Equilibrium, where the agents' decisions may be counterfactually dependent on each other, is how this relates to the correlated equilibrium.

A correlated equilibrium \citep{Aumann1974} \citep{Aumann1987} is a generalization (superset) of the Nash equilibrium. The agents receive, in advance, signals from a source, and these signals may be correlated. In practice, a signal often consists of a strategy profile that serves as a synchronization mechanism to ``agree'' on a specific Nash equilibrium.

A correlated equilibrium consists in a probability distribution on possible worlds\footnote{Please see Section \ref{section-possible-worlds} for a beginner's introduction to possible worlds, accessibility relations, the modeling of knowledge or belief, and partitions as a special case of accessibility relation when the relation is transitive, reflexive and symmetric.} (actually called states by the community), a partition of this set of possible worlds for each player that models their knowledge, and an assignment of a choice of strategy for each possible partition. It is an equilibrium if the expected utility is more than if the agents unilaterally modified their assignment of strategies to partitions.

From this definition, it is straightforward that correlated equilibria are part of the Nashian paradigm: changes of the assignments of strategies to partitions are done unilaterally, i.e., in spite of the correlation in the state of nature (which induces a correlation in the decisions, seen as random variables), the assignments of strategies are chosen independently of each other. Correlated equilibrium can be used to fine-tune cases in which there may exist several Nash equilibrium (e.g., in the Battle of the Sexes game), as the signal can be used to synchronize on one of the Nash equilibria.

\subsection{Quantum games (Nashian)}
\label{section-quantum-games}

In this paper, we are interested in games in normal form and look at pure strategies, i.e., players pick a single strategy. The Nash paradigm also allows mixed strategies, in which players may instead pick a probability distribution over their sets of strategies. In correlated equilibria, these probability distributions may be built on top of (possibly correlated) signals received by all agents before the game, which allows synchronization.

Quantum games \cite{Meyer1999} \cite{Benjamin2001} are a further extension this paradigm, in which players receive the signal from nature as a quantum state, and can also send their choice of strategy as quantum states (qubits) rather than classical states (i.e., classical probability distributions over sequences of bits for mixed strategies). There is a considerable body of literature on Bell inequalities \cite{Bell1964}\cite{Colbeck2017} that shows that the expressive power of qubits is strictly greater than that of classical bits, because of entanglements in the received signals (also quantum states) that allow breaking constraints that limit what classical bits can do.

The choices of strategies in quantum games by the agents correspond on the physical level to the application of unitary (reversible) transformations to their input quantum states, to the free choice of measurement axes, and to carrying out the corresponding measurements. Again, these choices, especially picking measurement axes, are again governed in this paradigm by the assumption that they are made fully independently from anything that could not have been caused by them\footnote{This elegant formal formulation is by \citep{Renner2011}.}.

We are actively investigating an alternative avenue of research, available publicly as pre-prints, in which we drop this assumption of independence by applying NNDT to quantum theory. For this, we generalized the PTE to decisions made in special-relativistic spacetime, which are shown to be expressible as games in extensive form with imperfect information \citep{Fourny2019a}. We then reformulated measurements and experiments as games played across spacetime between agents (physicists) and nature, which also maximizes its utility (this is known as the Principle of Least Action) \citep{Fourny2019b}. This leads to completely deterministic models based on pure strategies that offer the potential to extend quantum theory to a fully deterministic (and falsifiable) theory without being constrained by the Bell inequalities and without contradicting theoretical impossibility theorems\footnote{These theorems all rely on Nashian models of decision making. The most stringent such impossibility theorem currently known to us, by \cite{Renner2011} shows that quantum theory is maximally informative (and thus nature is inherently random) for agents endowed with free choice in the Nash sense (unilateral deviations).}, with the long-term goal of yielding concrete experimental protocols and possibly more powerful, fixed-point-based computational models \citep{Aaronson2008}. Because we use NNDT and not CDT, this approach is different from that taken in quantum games, which remains Nashian in nature as far as the agents' decisions on measurement axes are concerned.

\subsection{Welfare economics and social utility}
\label{section-social-welfare}

Social welfare functions \citep{Sen2018} are used to quantify the desirability of a society, which introduces a moral dimension to rational behavior. A particular category of functions take individual agent utilities and quantify how efficient and well-distributed their utilities are. Some examples of such functions are: the minimum utility, the sum or average of utilities, as well as the more elaborate \cite{Foster1983}'s functions, which also considers how well utilities are distributed. To this goal, the Theil L index \citep{Theil1996} is used to measure inequality, which is zero when all agents have the same utilities and highest when a single agent has everything; namely, for a specific game outcome $\sigma$:

$$T_L(\sigma)=\frac{1}{N} \sum_{i\in P} \log \frac{ \frac{1}{|P|} \sum_{j\in P} u_j(\sigma) }{u_i(\sigma)}$$

Foster's social welfare function is then given, for this same outcome $\sigma$, with:

$$W_{Foster}(\sigma)= \frac{1}{|P|} \sum_{i\in P} u_i(\sigma) \times e^{-T_L(\sigma)}$$

\begin{figure}
\centering{
\resizebox{0.7\textwidth}{!}{
\includegraphics{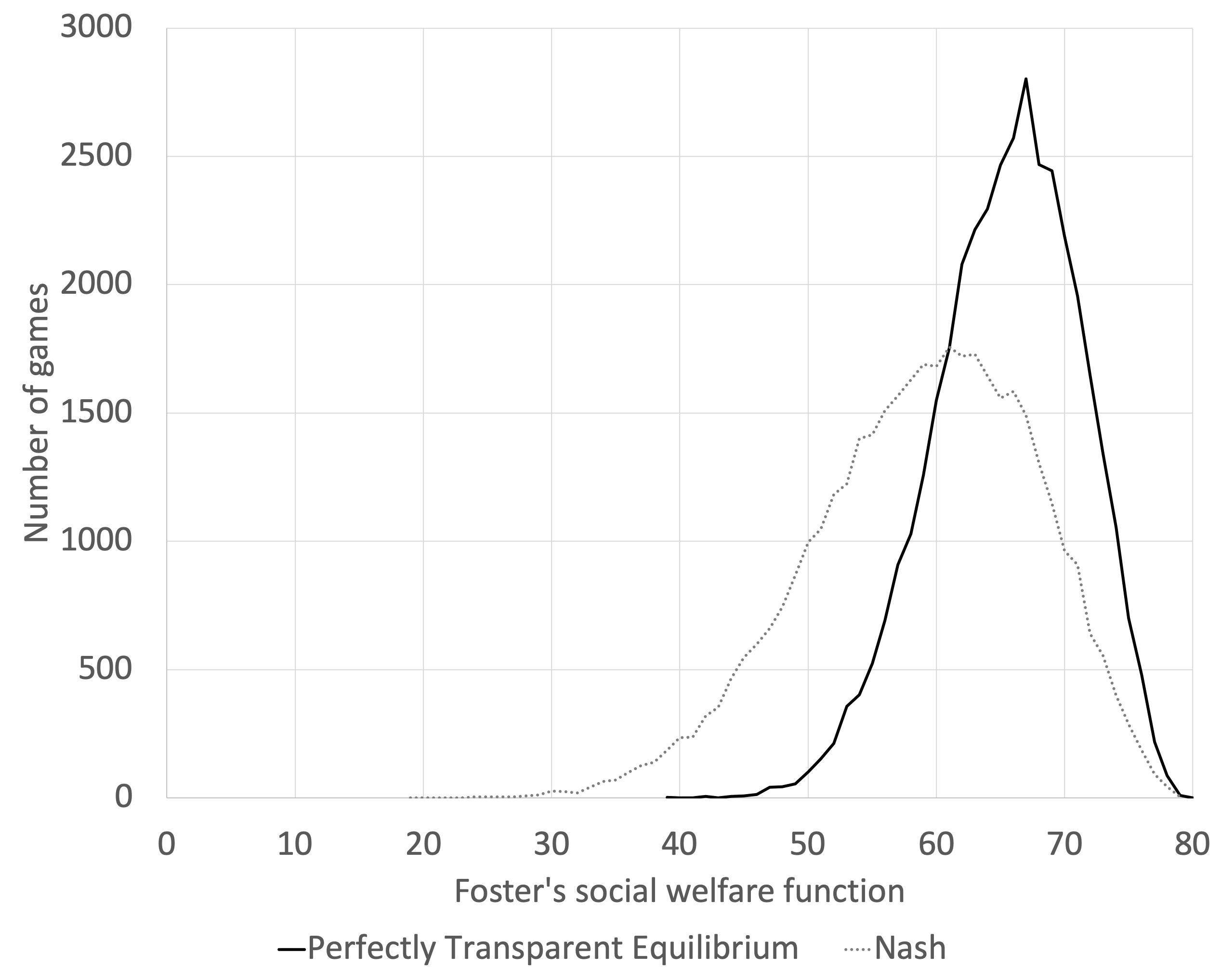}
}
}
\caption{Distribution of social welfare on randomly generated games involving 4 agents having each 3 strategies (ordinal preferences are randomly permuted), for both the Nash equilibrium (when it is unique) and the Perfectly Transparent Equilibrium. We use Foster's function with the Theil L index}
\label{fig-social-welfare}
\end{figure}

The Pareto optimality of the PTE is an indicator of its efficiency, but it does not take into account how well the payoffs are distributed. In order to understand how the PTE performs at the latter, we randomly generated games with 4 agents having each 3 strategies and computed the distribution of the values of Foster's social welfare function on the Perfectly Transparent Equilibrium when it exists on the one hand, and on the Nash equilibrium when it exists and is unique on the other hand. The resulting distributions are shown in Figure \ref{fig-social-welfare}.

\subsection{Evolutionary Game Theory}

Another line of research aiming at understanding altruistic behavior in nature is evolutionary game theory \citep{Sigmund2011}. The main difference with classical game theory is that agents were born with their strategies and transmit them to their offspring. At each iteration, the strategies are applied to obtain payoffs, which in turn affect how much each agent will reproduce.

The key difference with the PTE, and more generally the NNDT approach, is that evolutionary game theory is based on repeated games with a causal flow, while our approach is based on a one-off play. For games in normal form, in the non-Nashian reasoning, the strategies are thus counterfactually interdependent, but causally independent. The consequence of this difference is that evolutionary stable states are those that are individually rational (Folk theorem). This broader scope makes evolutionary game theory a good descriptive theory, while the counterfactuals-based reasoning structure of the PTE uniquely identifies at most a specific outcome. Today, the descriptive value of the PTE is thus inherently more limited than that of evolutionary game theory.

\subsection{Berge Equilibrium (Nashian)}

\cite{Berge1957} introduced a solution concept where agents are altruistic: rather than seeking to increase their utility, they seek to increase the other agents' utility. A follow-up on its properties was made by \cite{Colman2011}. Whereas in a Nash equilibrium, each agent best replies to the other players' (fixed) strategies, in a Berge Equilibrium, given the fixed strategy of an agent, a deviation of the \emph{other} agents' strategies never increases this agent's utility.

If the Berge equilibrium is also individually rational, it is called a Berge-Vaisman equilibrium. Since this is not always the case, it follows that a Berge equilibrium is not always a PTE. Furthermore, there may be multiple Berge equilibria. The Berge equilibrium assumes altruistic behavior, while the PTE assumes selfish behavior; the assumption of Perfect Prediction acts as an invisible hand that leads to Pareto optimality and higher social welfare than Nash equilibria in the long-term.

\section{The special case of symmetric games}
\label{section-symmetric-games}

There is a subclass of games that is of particular relevance, because players are interchangeable: symmetric games. Symmetric games are the category of games for which the earliest accounts of non-Nashian reasonings are found in the literature \cite{Hofstadter1983}. They are thus of particular importance for comparisons with more recent non-Nashian results such as minimax-rationalizability and the PTE, but also with the older concept of individual rationality (Folk theorem).

\subsection{Symmetric games}

In a symmetric game, the strategy spaces are identical for all players, and the payoffs are defined in such a way that the game is invariant through a permutation of players.

\begin{definition}[symmetric game] A game is symmetric if

\begin{itemize}

\item the strategy spaces are identical

$$\forall i, j \in P, \Sigma_i = \Sigma_j = \Upsilon$$

\item the payoffs are symmetric

$$\forall \pi \in Sym[P], \forall \sigma \in \Sigma, u_i(\sigma) = u_{\pi(i)}(\sigma_{\pi(.)})$$

where $Sym[P]$ is the permutation group on P and $\sigma_{\pi(.)}$ denotes $(\sigma_{\pi(1)},$ $\sigma_{\pi(2)}, ..., \sigma_{\pi(|P|)})$.

\end{itemize}

\end{definition}

\begin{figure}
\begin{center}
\begin{tabular}{|r|c|c|}
\hline
& Straight & Swerve\\
\hline
Straight & 0, 0 &\cellcolor{black!75}\textcolor{white}{3, 1}\\
\hline
Swerve & \cellcolor{black!75}\textcolor{white}{1, 3} & \cellcolor{black!25}2, 2\\
\hline
\end{tabular}
\end{center}
\caption{The chicken game. A player can either stay straight or swerve. If both swerve, they get more (aka lose less) than if they both stay straight, and a player who unilaterally goes straight gets more payoff than if both swerve. The difference with the prisoner's dilemma, however, is that the "betrayed" player has interest in not reciprocating the betrayal (0 and 1 are swapped). This game has two Nash equilibria: when players make opposite decisions.The individually rational outcomes are all those that Pareto-dominate the maximin tuple (1,1), that is, all but Straight-Straight.}
\label{fig-chicken-game}
\end{figure}

\begin{figure}
\begin{center}
\begin{tabular}{|r|c|c|}
\hline
& Sushi & Pizza\\
\hline
Sushi &  \cellcolor{black!75}\textcolor{white}{1, 1}&  \cellcolor{black!25}0,0\\
\hline
Pizza &  \cellcolor{black!25}0,0 &  \cellcolor{black!75}\textcolor{white}{2,2}\\
\hline
\end{tabular}
\end{center}
\caption{The coordination game. In this game, the players have a mutual interest to pick the same strategy, even though one of the two strategies is better for both of them (aligned interest). All diagonal outcomes are Nash equilibria, i.e., players will not deviate if they made the same decision.The individually rational outcomes are all those that Pareto-dominate the maximin tuple (0,0), that is all of them.}
\label{fig-coordination-game}
\end{figure}

\begin{figure}
\begin{center}
\begin{tabular}{|r|c|c|c|}
\hline
& A & B & C\\
\hline
A & \cellcolor{black!75}\textcolor{white}{9, 9}& \cellcolor{black!25}8,6 & 5,1\\
\hline
B &\cellcolor{black!25} 6,8 & \cellcolor{black!25}7,7 & 4,2\\
\hline
C & 1,5 & 2,4 & 3,3\\
\hline
\end{tabular}
\end{center}
\caption{A larger game. The Nash equilibrium is AD, i.e., (9,9). The individually rational outcomes are all those that Pareto-dominate the maximin tuple (5,5). Indeed, strategy A guarantees for both players a minimum payoff of 5 regardless of what the other does.}
\label{fig-minimax-individual}
\end{figure}

The prisoner dilemma (Figure \ref{fig-prisoner-dilemma}) is the most prominent symmetric game, found in almost any textbook of game theory. Other prominent symmetric game examples are the Chicken game (Figure \ref{fig-chicken-game}), the Coordination game (Figure \ref{fig-coordination-game}). Figure \ref{fig-minimax-individual} finally shows an example of a 3x3 game in general position. On all these figures, the Nash equilibrium and individually rational outcomes (which include Nash equilibria) are shown in black resp. gray.

\subsection{Superrationality}
\label{section-superrationality}

Superrationality was introduced by Douglas Hofstadter in 1983 for symmetric strategic games in a Scientific American column. Hofstadter's argument was made in the introduction (Section \ref{section-superrational-thinking}), where we quoted one of his most brilliant explanations. Douglas Hofstadter directly put in question the fundamental assumption behind Nash equilibria that players consider their decisions to be counterfactually independent of other players' decisions. Superrational players consider that their reasonings are interdependent, not because of any causal effect or any kind of retrocausality, but because their reasonings and conclusions are identical. As argued in Section \ref{section-newcomb}, an agent's decision may not be mere evidence of the other agent's decision (this is what Evidential Decision Theory argues): there can be an actual counterfactual implication between the agents' decisions under a weakened free choice assumption (a form of reasoning that we called Non-Nashian Decision Theory).

The games described in Hofstadter's column are all symmetric, which is a requirement for the reasonings to be identical. Identical reasonings and conclusions mean that only outcomes on the diagonal of the normal form are considered. An equilibrium is reached if, among all outcomes of the diagonal, it leads to the highest payoffs (which does not depend on the player as the game is symmetric).\footnote{Note that the original paper by Douglas Hofstadter defines the equilibrium reached by superrational players for variants of the prisoner's dilemma. We are taking the liberty of formally extending the reasoning on any symmetric games. However, it has to be said that we are unsure of what Douglas Hofstadter would think of games with ties on the diagonal, as there would not be a unique solution.}

Formally, this is expressed as follows.

\begin{definition}[Hofstadter equilibrium] Given a symmetric game $\Gamma=(P, (\Sigma_i)_i, (u_i)_i)$ in normal form, a strategy profile $\overrightarrow\sigma$ is an equilibrium reached by Superrational players (a Hofstadter equilibrium) if:

\begin{itemize}

\item the strategy profile is on the diagonal:

$$\exists\upsilon\in\Upsilon, \overrightarrow\sigma=(\upsilon, \upsilon, ..., \upsilon)$$

which we can also express as

$$\overrightarrow\sigma\in diag(\Sigma)$$

\item it maximizes the payoff on the diagonal

$$\forall \overrightarrow\tau \in diag(\Sigma), \forall i \in P, u_i(\overrightarrow\sigma) \ge u_i(\overrightarrow\tau)$$

\end{itemize}

\end{definition}

We are unsure of how Hofstadter would describe Superrational behavior in games with ties, as several equilibria may emerge, contradicting the premise of the reasoning. However, in this paper, we assume that games have no ties so that the Hofstadter equilibrium exists and is unique.

Figures \ref{fig-prisoner-dilemma-hofstadter}, \ref{fig-chicken-game-hofstadter} and \ref{fig-coordination-game-hofstadter} show the Hofstadter equilibria for our example games.

\begin{figure}
\begin{center}
\begin{tabular}{|r|c|c|}
\hline
& Defect & Cooperate\\
\hline
Defect & 1, 1 & 3, 0\\
\hline
Cooperate &  0, 3 & \cellcolor{black!25}2, 2\\
\hline
\end{tabular}
\end{center}
\caption{The prisoner's dilemma. Superrational players either both cooperate or both deviate. In a Hofstadter equilibrium, players both cooperate.}
\label{fig-prisoner-dilemma-hofstadter}
\end{figure}

\begin{figure}
\begin{center}
\begin{tabular}{|r|c|c|}
\hline
& Straight & Swerve\\
\hline
Straight & 0, 0 & 3, 1\\
\hline
Swerve & 1, 3 & \cellcolor{black!25}2, 2\\
\hline
\end{tabular}
\end{center}
\caption{The chicken game. Superrational players either both stay straight or swerve. In a Hofstadter equilibrium, players both swerve.}
\label{fig-chicken-game-hofstadter}
\end{figure}

\begin{figure}
\begin{center}
\begin{tabular}{|r|c|c|}
\hline
& Sushi & Pizza\\
\hline
Sushi & 1, 1& 0,0\\
\hline
Pizza & 0,0 & \cellcolor{black!25}2,2\\
\hline
\end{tabular}
\end{center}
\caption{The coordination game. Superrational players either both pick Sushi or Pizza. In a Hofstadter equilibrium, players both pick Pizza.}
\label{fig-coordination-game-hofstadter}
\end{figure}

An alternative formalization of Hofstadter's superrationality on symmetric games is given by \cite{Thome2019}, called superrational types. Types model beliefs that agents have on the types and decisions of other agents. Superrational types correspond to cases where the type and action of any agent are perfectly correlated with the types and actions the agent believes other agents have\footnote{The authors also analyze the case where agents may have different type spaces, in which case identifications are made to fall back to symmetric type spaces.}. In this respect, to the best of our understanding, this approach is explicitly Bayesian, does not rely on counterfactuals, and is an EDT account of superrationality, whereas our approach is Non-Nashian and based on (possibly non-symmetric) probabilities of subjunctive conditionals (counterfactuals). This paper shows that an NNDT approach is key to extending superrationality to non-symmetric games.

\subsection{Inclusion theorems specific to symmetric games}

We now turn to inclusion theorems involving Hofstadter's Superrationality, minimax-rationalizability, individual rationality, and the Perfectly Transparent Equilibrium. These theorems are proven here, and the proofs are relatively succinct. These results were also confirmed experimentally on a very large quantity of games of various sizes by Felipe \cite{Sulser2019}, with the datasets publicly available online, including the game configurations annotated with their resolution to various equilibrium concepts.

We give the inclusion theorems in this order, which is summarized in Figure \ref{fig-venn}.

\begin{itemize}
\item PTE $\subset$ Hofstadter equilibrium
\item Hofstadter equilibrium $\subset$ minimax-rationalizability
\item Hofstadter equilibrium $\subset$ individual rationality
\end{itemize}

Interestingly, it follows by transitivity that, on symmetric games, the PTE, when it exists, is always minimax-rationalizable. This is not true in general for asymmetric games, as we will see in Section \ref{section-counterexamples}

\begin{figure}
\centering{
\resizebox{0.7\textwidth}{!}{
\includegraphics{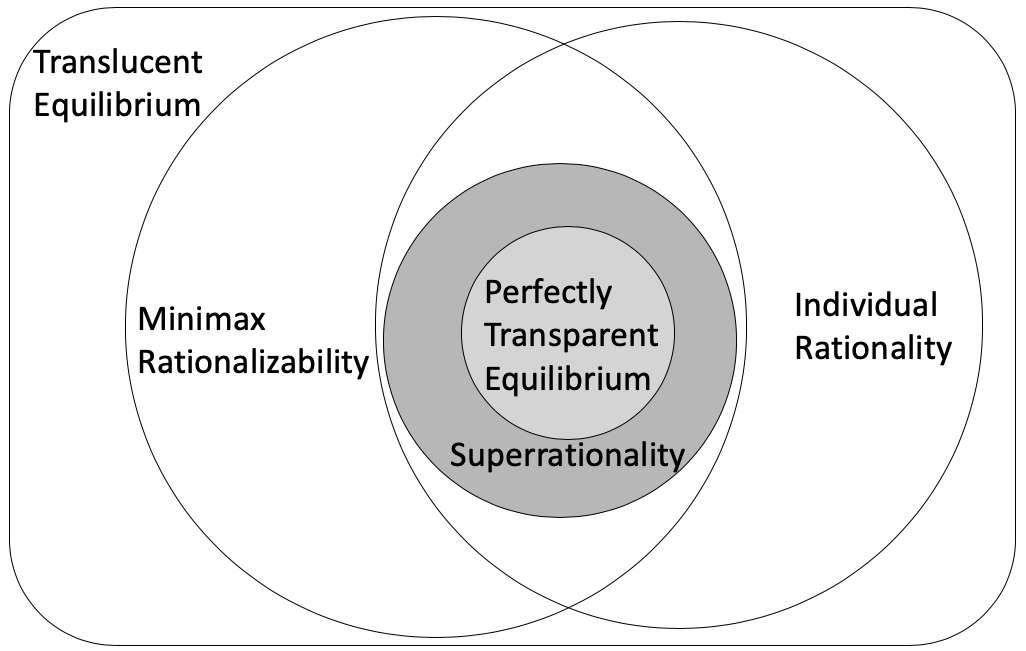}
}
}
\caption{A diagram depicting the relationship between the Translucent Equilibrium, Minimax Rationalizability, Individual Rationality, Superrationality and the Perfectly Transparent Equilibrium for strategy profiles on symmetric games.}
\label{fig-venn}
\end{figure}

\begin{theorem}[PTE $\subset$ Hofstadter equilibrium]
Given a symmetric game in normal form, with pure strategies, and with no ties, if the Perfect Prediction Equilibrium exists, then it is identical to Hofstadter's equilibrium.
\end{theorem}

\begin{proof}[PTE $\subset$ Hofstadter equilibrium]
Assume the PTE exists for a given game. In a symmetric game, the PTE must be on the diagonal. This is because, if the PTE were not on the diagonal, one would obtain another distinct PTE by swapping the role of the players. This would contradict uniqueness. Thus, the PTE lies on the diagonal. Since it must be Pareto-optimal, it must correspond to the maximum diagonal payoff and coincide with the Hofstadter Equilibrium, as non-maximum diagonal payoffs are Pareto-dominated by the Hofstadter Equilibrium. $\square$
\end{proof}

Even though it reaches the same conclusion on symmetric games, the PTE tells a different story than Superrationality in its original form. Even though non-diagonal outcomes are all eliminated (not because of their ``non-diagonalness'', but for other arguments), the decision remains a strategic decision: cooperate or defect in the prisoner's dilemma. Given a decision (in some given possible world), the payoff obtained must be compared to the counterfactual payoffs obtained if the other decision had been made (which it could, counterfactually). This counterfactual structure leads to the conclusion that cooperating is the rational choice ($2>1$) in the prisoner's dilemma.

\begin{theorem}[Hofstadter equilibrium $\subset$ minimax-rationalizability]
Given a symmetric game in normal form, a Hofstadter equilibrium is always minimax-rationalizable.
\end{theorem}

\begin{proof}[Hofstadter equilibrium $\subset$ minimax-rationalizability]

In minimax rationalizability, the order of elimination is not relevant. Because of symmetry, if a strategy gets eliminated for a player, then it will be eliminated for all players. We reorder eliminations in such a way that strategies get eliminated for all players in batches so that after each elimination, the game remains symmetric. We can now show that, for a symmetric game, a Hofstadter equilibrium cannot get minimax-eliminated.

Let $\overrightarrow\sigma$ be a Hofstadter equilibrium. We can write it as $\overrightarrow\sigma=(\sigma, \sigma, ..., \sigma)$ for some $\sigma\in\Upsilon$.

By definition of the maximum:

$$\max_{\tau_{-i}\in\Sigma_{-i}} u_i(\sigma, \tau_{-i}) \ge u_i(\sigma, \sigma, ..., \sigma)$$

because $(\sigma, ..., \sigma)$ is in the set over which the maximum is taken (opponents' strategies).

Let now $i$ denote a player, and $\upsilon\in\Upsilon$ now be any of its strategies. By definition of a Hofstadter equilibrium, the payoffs are maximal on the diagonal, so that:

$$u_i(\sigma, \sigma, ..., \sigma) \ge u_i(\upsilon, \upsilon, ..., \upsilon)$$

Finally, by definition of the minimum:

$$u_i(\upsilon, \upsilon, ..., \upsilon) \ge \min_{\tau_{-i}\in\Sigma_{-i}} u_i(\upsilon, \tau_{-i}) $$

because $(\sigma, ..., \sigma)$ is in the set over which the maximum is taken (opponents' strategies). By transitivity, we get:

$$\max_{\tau_{-i}\in\Sigma_{-i}} u_i(\sigma, \tau_{-i}) \ge \min_{\tau_{-i}\in\Sigma_{-i}} u_i(\upsilon, \tau_{-i})$$

which directly contradicts the existence of a strategy that allows minimax-domination, and this holds for any player.

There is one more thing to say for the proof to be complete. After an iteration of the deletion of minimax-dominated strategies as described above, a Hofstadter equilibrium remains a Hofstadter equilibrium. This is because eliminating other rows or columns than that of the maximum diagonal payoff does not affect this maximum diagonal payoff. Hence, a Hofstadter equilibrium will recursively survive all iterations and, in the end, satisfy minimax rationalizability. $\square$
\end{proof}

\begin{theorem}[Hofstadter equilibrium $\subset$ individual rationality]
Given a symmetric game in normal form, a Hofstadter equilibrium is always individually rational.
\end{theorem}

\begin{proof}[Hofstadter equilibrium $\subset$ individual rationality]

Let $\overrightarrow\sigma$ be a Hofstadter equilibrium. We can write it as $\overrightarrow\sigma=(\sigma, \sigma, ..., \sigma)$ for some $\sigma\in\Upsilon$.
Let $i$ denote a player. By definition of a Hofstadter equilibrium, the payoffs are maximal on the diagonal, so that:

$$u_i(\sigma, \sigma, ..., \sigma) \ge \max_{\upsilon\in\Upsilon} u_i(\upsilon, \upsilon, ..., \upsilon)$$

Furthermore, for any strategy $\upsilon$, 

$$u_i(\upsilon, \upsilon, ..., \upsilon) \ge \min_{\tau_{-i}\in\Sigma_{-i}}  u_i(\upsilon, \tau_{-i})$$

(the minimum payoff on its line can only be smaller than the payoff on the diagonal). Combining the above inequalities:

$$u_i(\sigma, \sigma, ..., \sigma) \ge \max_{\upsilon\in\Upsilon} u_i(\upsilon, \upsilon, ..., \upsilon) \ge \max_{\upsilon\in\Upsilon} \min_{\tau_{-i}\in\Sigma_{-i}}  u_i(\upsilon, \tau_{-i})$$

Considering that $\upsilon$ is a mute variable and that $\Upsilon=\Sigma_i$ (symmetric game),

$$u_i(\sigma, \sigma, ..., \sigma) \ge  \max_{\tau_i\in\Sigma_i} \min_{\tau_{-i}\in\Sigma_{-i}}  u_i(\tau_i, \tau_{-i})$$

this fulfils the definition of an individually rational outcome. $\square$
\end{proof}

\section{Counterexamples}
\label{section-counterexamples}

We now come back to games that are potentially asymmetric. We already know that the PTE is always individually rational by definition, and we can also establish that some other inclusions are not true in general with a few counterexamples. The general inclusion diagram of the PTE, minimax-rationalizability and individual rationality is shown in Figure \ref{fig-venn2}.

\begin{figure}
\centering{
\resizebox{0.7\textwidth}{!}{
\includegraphics{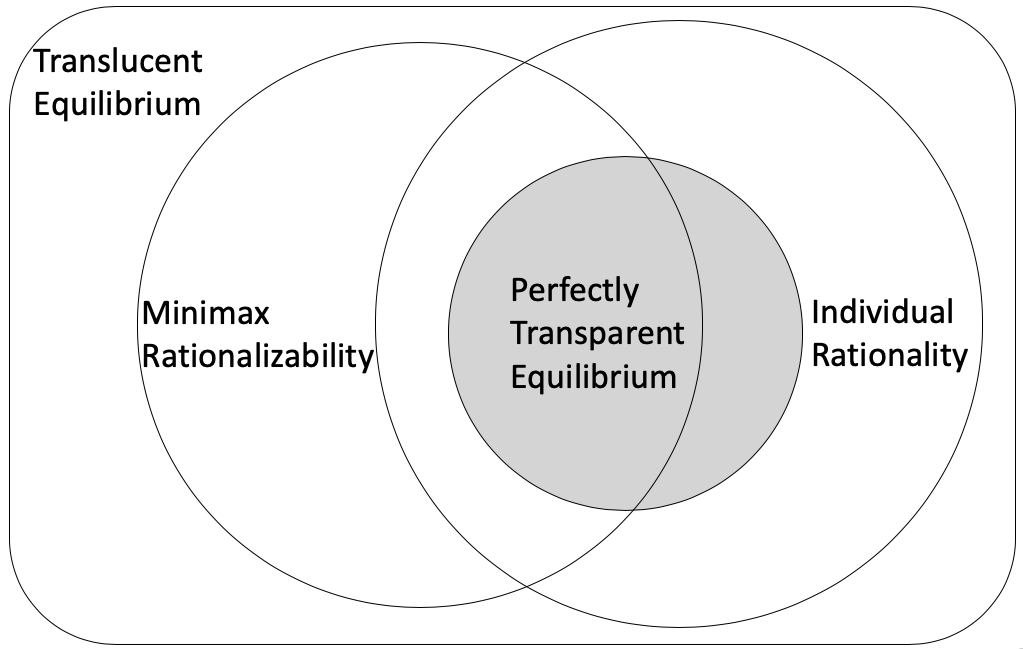}
}
}
\caption{A diagram depicting the relationship between the Translucent Equilibrium, Minimax Rationalizability, Individual Rationality and the Perfectly Transparent Equilibrium for strategy profiles on non-symmetric games: one of the inclusions does not hold in general.}
\label{fig-venn2}
\end{figure}

\subsection{Game with no PTE}

The PTE does not always exist, and the Chicken game provides a counter-example. Figure \ref{fig-example3} gives the detail of the reasoning on the chicken game. In the first round, the maximin utilities are (1,1). Strategy profile AC is both preempted by B and by D. In the second round, the new maximin utilities are (3,3). None of the remaining strategy profiles Pareto-dominates (3,3), so that none of them is stable: no matter which of the remaining outcomes would actually happen, at least one agent would not have acted rationally under the assumptions of Necessary Rationality and Necessary Knowledge of Strategies.

\begin{figure}
\begin{center}

\resizebox{\textwidth}{!}{
\begin{tabular}{lcccr}

\begin{tabular}{|r|c|c|}
\hline
& A & B \\
\hline
C & 0, 0& 3, 1\\
\hline
D & 1, 3& 2,2\\
\hline
\end{tabular}

&

$\Rightarrow$

&

\begin{tabular}{|r|c|c|}
\hline
& A & B \\
\hline
C & \cellcolor{black!25}0, 0& 3, 1\\
\hline
D & 1, 3& 2,2\\
\hline
\end{tabular}

&

$\Rightarrow$

&

\begin{tabular}{|r|c|c|}
\hline
& A & B \\
\hline
C & \cellcolor{black!25}0, 0& \cellcolor{black!25}3, 1\\
\hline
D & \cellcolor{black!25}1, 3& \cellcolor{black!25}2,2\\
\hline
\end{tabular}

\end{tabular}
}
\end{center}
\caption{Iterated elimination of preempted strategy profiles in the the chicken game. The tuples of maximin utilities are (1,1), then (3,3). No strategy profile remains: there is no PTE.}
\label{fig-example3}
\end{figure}

The chicken game thus has no PTE. Another example of a game with no PTE is the Battle of the Sexes game\footnote{we refer to its variant in general position.}. Even though there is no PTE for all games, our framework is not silent on these games: it does say that there exists no equilibrium under its underlying assumptions. This is an improvement over the original superrationality paradigm, which remains completely silent on symmetric games because it relies on the symmetry argument.

\subsection{Symmetric game with a Hofstadter equilibrium which is not a PTE}

The game shown in Figure \ref{fig-example3} shows a Hofstadter equilibrium (2, 2) that is not a PTE: it is not immune to Necessary Rationality and Necessary Knowledge of Strategies, because the row player would have rationally preferred C (getting $3>2$ knowing that CA is impossible) to B, and, likewise, the column player would have rationally preferred B to D knowing that AC is impossible. This proves that Necessary Rationality and Necessary Knowledge of Strategies is a slightly more stringent assumption than Superrationality.

\subsection{Game with a PTE that is not rationalizable}

Figure \ref{figure-asymmetric-social-dilemma} shows a counterexample in which the PTE is not rationalizable. Indeed, strategy F is a dominant strategy, i.e., it always the best response of the row player. However, the PTE is (D, A). This is due to rationalizability being specific to the Nash paradigm, in which deviations are unilateral only.

\subsection{Game with a PTE that is not a Nash equilibrium}

Figure \ref{figure-asymmetric-social-dilemma} shows a counterexample in which the PTE is not a Nash equilibrium. The prisoner's dilemma is also such a game. For some other games, the PTE is a Nash equilibrium -- for example, a game where both players always get the same payoffs and have aligned interests. The non-Nashian paradigm, in general, is thus in no particular inclusion relationship with Nash equilibria.

\subsection{Game with a PTE that is not minimax-rationalizable}

\label{section-counterexample-minimax}

While minimax-rationalizability is a non-Nashian concept, and while the PTE on symmetric games is always minimax-rationalizable as shown in Section \ref{section-symmetric-games}, the PTE is not always minimax-rationalizable in general, showing a singularity between translucency and full transparency.

\section{Motivation and practical use}

The two concepts of Necessary Rationality and Necessary Knowledge of Strategies can be seen as strong assumptions. In practice, however, these assumptions can be reduced to the agents' \emph{believing} that they are rational in all possible worlds and correctly predict each other in all possible worlds, as pointed out by \cite{Dupuy2000}. Whether or not these beliefs are actually correct is irrelevant, and the Perfectly Transparent Equilibrium is reached even with this weaker assumption: in the end, the decisions made by agents are driven by their own beliefs and rational reasoning under these beliefs.

Having this in mind, we see six relevant domains of applicability of the Perfectly Transparent Equilibrium:

\subsection{Philosophy of ethics}

First, in some situations including the prisoner's dilemma, the behavior predicted by the Perfectly Transparent Equilibrium, e.g., cooperating, is actually encountered in the real world for some agents. The belief in Necessary Rationality and Necessary Knowledge of Strategies can be seen as a formal model for \emph{describing} honest behavior: somebody who holds the sincere belief that they are an open book and that everybody else correctly anticipates all their actions is less likely to choose to betray other agents -- regardless of whether this belief is correct or not. As of today, the descriptive value of the PTE is limited to such very simple situations; its counterpart in extensive form, the Perfect Prediction Equilibrium, also has descriptive value for asynchronous exchange settings.

However, Necessary Rationality and Necessary Knowledge of Strategies also provide a \emph{prescriptive} moral framework that helps to define ethical behavior. Dupuy's work is largely inspired by Kantian philosophy, where acting based on this belief can be packaged as a categorical imperative to do so.

For games in extensive form \citep{Fourny2018}, this Kantian imperative seeks to avoid inconsistencies and seek a consistent timeline where the prediction of the solution of the game causes that solution to be reached, as a self-fulfilling prophecy. In Dupuy's words, ``never act in such a way that, had your action been anticipated, it would not be in your power to carry it out.``

We saw in Section \ref{section-social-welfare} that, on games with randomly picked permutations of preferences (in our example 4 agents, 3 strategies), the PTE leads on average to increased social welfare, which in turns advocates for a normative value of non-Nashian thinking.

\subsection{Bargaining}
Nashian agents that would, under the classical Nash reasoning, reach a Nash equilibrium that is sub-optimal, such as both betraying in the prisoner's dilemma (see also ``social dilemma' examples in Section \ref{section-social-dilemma}) often realize that they are stuck in an unsatisfactory equilibrium. They can choose to enter a bargain and commit contractually to different choices that are otherwise inaccessible to Nashian agents. The Perfectly Transparent Equilibrium provides a natural basis for such a bargain: acknowledging that the shared belief in Necessary Rationality and Necessary Knowledge of Strategies always leads to Pareto-optimal outcomes, the agents can agree to commit to act ``as if'' they believed so, and, draft the choice of strategies into a binding contract accordingly. This is a pragmatic and practical way to find and agree on a Pareto-optimal outcome, and also to increase in the long-term social welfare.

\subsection{Predicting human decisions with Machine Learning}
\label{section-ml-usecase}

A more long-term motivation for the Perfectly Transparent Equilibrium, and the non-Nashian approach in general, is that Machine Learning techniques are constantly getting better at predicting human behavior. An example thereof is the work by \cite{Kadar2015}, which predicts in advance where burglaries are most likely to happen so that the police can optimize their patrols.

Predicting human behavior, however, is different from predicting tomorrow's weather. Indeed, when the agents that are being predicted are informed of the prediction (valid in one possible world) in advance, they can adapt their behavior and make this prediction false. This is addressed in the context of CDT by behavioral game theory \citep{Allais1953}. An approach constrained by CDT, however, has fundamental limitations (this is also true for EDT), which are discussed at a considerable level of depth by \cite{Dupuy1992}. Dupuy seeded the counterfactual approach by observing that CDT is based on predictions being true in the actual world and yet counterfactually falsifiable by making a different decision, while our direction of research assumes they are necessarily true, because counterfactually dependent on the decision.

Any system that claims to predict the decisions made by human beings must thus consider its own impact on the agents it predicts, i.e., the fact that people know in advance what is being predicted that they will do. Thus, perfect prediction can only be achieved by considering this as a fixpoint problem, taking the implications of its own anticipation into account.

As of today, our machine learning technologies are not precise enough for people to be convinced that they are being perfectly predicted. The descriptive value of the PTE is thus limited today, and we will stick, for the purpose of this paper, to a few conjectures that may or may not prove correct.

We expect that, in the coming decades, as AI systems become more powerful and accurate, in some contexts, people will begin factoring in and completely trusting that, no matter what they decide to do, the AI system will have predicted it. We thus conjecture that the descriptive value of the PTE will increase and account for a growing number of decision-making situations.

If this conjecture is realized, we suspect that there may be a transition period in which the level of accuracy of such systems will not be high enough for reasoning in terms of perfect prediction, but still high enough for deviations not to occur unilaterally. In other words, we may enter a gray zone in which the Nash paradigm slowly loses its descriptive value, but in which the perfect prediction paradigm does not yet have a descriptive value; for this, the translucency paradigm introduced by \cite{Halpern:2013aa} might provide a better account, but at the cost of equilibrium multiplicity.

\subsection{Ethics of AI and robots}

Leaving the domain of validity of the Nash approach to game theory may happen at a faster pace if people start delegating some of their decisions to AI recommendation systems \citep{Harari2015} as they realize that these systems know them better than themselves. Indeed, AI recommendation systems (e.g., movies on Netflix, books on Amazon, restaurants, dating websites, etc.) are purely based on algorithms and code. If Artificial Intelligence becomes a proxy for human decisions, the mutual strategic interactions shift from humans to machines, and optimizing decisions in this environment becomes a purely programmatic and algorithmic problem.

In a fully transparent setup in which the machines know one another's code and algorithms, the non-Nashian concepts of Necessary Rationality and Necessary Knowledge of Strategies become more adequate models than Nashian rationality and unilateral deviations.

Furthermore, due to the advantageous economic properties of the Perfectly Transparent Equilibrium as well as its extensive-form counterpart \citep{Fourny2018}, namely, their Pareto-optimality, having programming machines (e.g., smart contracts) that interact with each other transparently under Necessary Rationality and Necessary Knowledge of Strategies can lead to desirable, Pareto-optimal outcomes in general, and to Pareto-improvements over Nash equilibria in some (but not all) settings. This provides an additional incentive to manage the behavior of Artificial Intelligence systems interacting with each other so as to emulate what could be interpreted as honest behavior, with ramifications in the ethics of AI.

An interesting avenue of research is the study of the interaction of non-Nashian agents with Nashian agents, as the latter could, for example, hide their code. A Nashian agent has an incentive to exploit the algorithmic behavior of a non-Nashian agent. This in turns means that a non-Nashian algorithm should only behave as such when interacting with another non-Nashian agent with the same algorithm, and otherwise adapt, for example, by mimicking Nashian behavior. Related work on this meta-level of thinking includes extended preferences \citep{Greaves2018}, in which agents can have ``preferences on their preferences'', although the case of what one could call ``preferences over beliefs'', ``preferences over decision theories'' or the dynamic switch between decision theories to adapt to one's opponent is likely to be more complex.

There are also numerous discussions in this respect related to Newcomb's problem: a two-boxer agent who is convinced that he was rational getting \$1,000 from his two boxes might nevertheless feel uneasy observing one-boxers effortlessly getting \$1,000,000, and wish he had been a one-boxer as well, without it being a formal contradiction in an extended framework.

\subsection{Modeling of an absent or impaired theory of mind}

The belief of Necessary Rationality and Necessary Knowledge of Strategies can be directly put into perspective with the absence or impairment of a theory of mind \citep{Premack1978} for an agent. ``Theory of mind'' refers to the ability to distinguish between one's knowledge and the knowledge of other agents.

Even though at first sight many would argue that they are more stringent assumptions than the simpler, Nashian assumptions of Common Knowledge of Rationality and Common Knowledge of Strategies, it could as well be argued, on the contrary, that these assumptions are, in fact, \emph{simpler assumptions} than in the Nash equilibrium reasoning. Indeed, agents in the Nash paradigm must keep track of who knows what, who knows who knows what, etc., which requires a significant amount of resources in terms of computation and memory from the brain. Higher levels of reasoning are even inaccessible to human agents in practice.

Necessary Rationality and Necessary Knowledge of Strategies flattens the structure of knowledge by idealizing agents as being epistemically omniscient. This also leads to less computational complexity in our algorithms, which is even more visible in the extensive form variant of the paradigm, which is based on a forward induction \citep{Fourny2018}. Children with a not-yet-developed or impaired theory of mind, for example, do not lie well or at all, simply put, because they believe that everybody else knows what they know \citep{Ding2015} \citep{Evans2013}  \citep{Baron1985}. The belief of Necessary Rationality and Necessary Knowledge of Strategies can thus be used as a descriptive model for such agents.

There is an interesting argument made by Rich \cite{Shiffrin2009}, who uses a thought experiment in which an agent is playing against themselves, but are taking Midalozam, a drug that erases short-term memory, between making their decisions. This can also be used in games in normal form, by having the agent make the row and column decisions in turn, knowing that they are making both decisions, although not jointly. The decisions made in such a setup still remain strategic decisions, making this process non-trivial.

\subsection{Deterministic extension of quantum theory}

Another promising use case for the line of research presented in this paper is the design of a model to extend quantum theory to a deterministic quantum theory, as suggested by \cite{Einstein1935}. A concrete model modeling quantum experiments, such as the EPR, as a game played between humans and the universe that can be solved for its Perfectly Transparent Equilibrium is given by \cite{Fourny2019b}.

As it turns out, the assumption that blocks deterministic extensions of quantum theory in impossibility theorems \citep{Renner2011} is the exact same assumption made in the Nash paradigm, i.e., that decisions made by physicists on what to measure are made independently. It is the very assumption that we as well as others (Halpern, Pass, Shiffrin, Hofstadter, Shiffrin...) are challenging, leading to an alternative paradigm we call Non-Nashian Decision Theory. More on this is said in Section \ref{section-quantum-games}.

\subsection{Descriptive or normative value of the PTE}

Is Non-Nashian Decision Theory, and more specifically the PTE, descriptive or normative? The PTE has some normative value in the sense that it can be used as the basis for a social norm for the decision making of individual agents that share this norm. The agents selfishly maximize their utility under this social norm, namely, that they all believe they are perfectly predictable. Then, an ``invisible hand'' similar to Adam Smith's view leads to the PTE, when it exists, and to its benefits in terms of Pareto-optimality (Section \ref{section-pareto-optimal}) and social welfare (Section \ref{section-social-welfare}).

This alternative form of decision-making provides an alternative to Nashian thinking but does not claim any form of monopoly over rationality. In \cite{Dupuy2007}'s words regarding the philosophical and ethical foundations of this thinking: ``The ruses based on the capacity to surprise are essential to human and social interactions, especially in the political arena. My claim is more modest. It does not consist in rejecting the kind of rationality proper to occurring time \footnote{The notion of time specific to Causal Decision Theory, but also to Evidential Decision Theory.} as embodied in orthodox Decision Theory\footnote{CDT.}. It insists that there exists another form of practical reason, no less important, that goes along with a different temporality and is associated with another kind of decision making.''

Today, the PTE has some descriptive value for very simple games and for some agents; indeed, there are many situations in the real world when people collaborate or act altruistically or honestly. The PTE, and NNDT in general, provide a framework that can describe this form of decision making and of its underlying social norm with alternative reasoning based on utility maximization. We emphasize, however, that as of today this descriptive value is limited and specific to certain agents and situations, while we also suggest (Section \ref{section-ml-usecase}) that this descriptive value might increase in the future. This suggestion is based, among others, on impossibility theorems in quantum physics that demonstrate, assuming the correctness of quantum theory's predictions, that Nature is contextual\footnote{This means, in simple terms, that what we see would have been different if our choice of the \emph{other} things that we decide to look at had been different. This property of Nature is fundamentally non-Nashian.} in essence \citep{Kochen1967}. This entails that CDT is an approximation of reality (a very good one for today's purposes) and that our ability to better describe Nature can only come at the cost of weakening the concept of independent decision making, which is what NNDT does.

\section{Conclusion}
\label{section-conclusion}

We introduced a new equilibrium for games in normal form in general positions reached under Necessary Rationality and Necessary Knowledge of Strategies: the Perfectly Transparent Equilibrium. We also emphasized the underlying counterfactual reasoning as being a decision theory, common to several non-Nashian results (Hofstadter, Dupuy, Halpern, us) that is neither CDT nor EDT: Non-Nashian Decision Theory.

In the case of symmetric games, we established inclusion relationships between the Hofstadter equilibrium, the Perfectly Transparent Equilibrium, minimax-rationalizability, and individual rationality. As mentioned in \citep{Fourny2018}, we suspected that the PPE was some kind of counterpart of the superrational thinking on games in extensive forms, and the PTE that we have just defined acts here as a missing link between the two.

The non-Nashian assumptions behind the PTE correspond to one-boxer behavior in Newcomb's Paradox, as \citet{Dupuy1992} showed. These assumptions describe an alternative form of rationality that explains the behavior of some agents that do not follow Nashian predictions. This form of rationality is based on the belief that decisions are correctly predicted in all possible worlds, and that agents are rational in all possible worlds.

\section{Acknowledgements}

I am first and foremost indebted to Jean-Pierre-Dupuy, who spent decades designing the decision theory framework of projected time and perfect prediction and formulated the initial conjectures, as well as to St\'ephane Reiche, with whom we collaborated on the formalism of the PPE in extensive form. Some counterexamples in this paper were found by Felipe Sulser during his Master's thesis, focused on a bootstrap of a cartography of non-Nashian game theory based on the analysis of game datasets. Gustavo Alonso co-supervised this thesis and also supported me with the costs of running the large-scale experiments on a data center infrastructure. I am also thankful to Alexei Grinbaum, Bernard Walliser, Rich Shiffrin, Bob French, Joe Halpern for exciting discussions on the topic, as well to Elliott Ash and Li Jialin for exciting discussions as well as proof-reading the paper.

\bibliographystyle{spbasic}      

\begin{thebibliography}{54}
\providecommand{\natexlab}[1]{#1}
\providecommand{\url}[1]{{#1}}
\providecommand{\urlprefix}{URL }
\expandafter\ifx\csname urlstyle\endcsname\relax
  \providecommand{\doi}[1]{DOI~\discretionary{}{}{}#1}\else
  \providecommand{\doi}{DOI~\discretionary{}{}{}\begingroup
  \urlstyle{rm}\Url}\fi
\providecommand{\eprint}[2][]{\url{#2}}

\bibitem[{Aaronson and Watrous(2008)}]{Aaronson2008}
Aaronson S, Watrous J (2008) Closed timelike curves make quantum and classical
  computing equivalent. Proceedings of the Royal Society A: Mathematical,
  Physical and Engineering Sciences 465(2102):631--647

\bibitem[{Allais(1953)}]{Allais1953}
Allais M (1953) Le comportement de l'homme rationnel devant le risque: critique
  des postulats et axiomes de l'{\'e}cole am{\'e}ricaine. Econometrica: Journal
  of the Econometric Society pp 503--546

\bibitem[{Aumann(1974)}]{Aumann1974}
Aumann RJ (1974) Subjectivity and correlation in randomized strategies. Journal
  of mathematical Economics 1(1):67--96

\bibitem[{Aumann(1987)}]{Aumann1987}
Aumann RJ (1987) Correlated equilibrium as an expression of bayesian
  rationality. Econometrica: Journal of the Econometric Society pp 1--18

\bibitem[{Baron-Cohen et~al(1985)Baron-Cohen, Leslie, and Frith}]{Baron1985}
Baron-Cohen S, Leslie AM, Frith U (1985) Does the autistic child have a
  ``theory of mind''? Cognition 21(1):37--46

\bibitem[{Bell(1964)}]{Bell1964}
Bell J (1964) {{On the Einstein Podolsky Rosen Paradox}}. Physics 1(3):195--200

\bibitem[{Benjamin and Hayden(2001)}]{Benjamin2001}
Benjamin SC, Hayden PM (2001) Multiplayer quantum games. Physical Review A
  64(3):030,301

\bibitem[{Berge(1957)}]{Berge1957}
Berge C (1957) Th{\'e}orie g{\'e}n{\'e}rale des jeux {\`a} n personnes, vol
  138. Gauthier-Villars Paris

\bibitem[{Bil{\`o} and Flammini(2011)}]{bilo2011}
Bil{\`o} V, Flammini M (2011) Extending the notion of rationality of selfish
  agents: Second order nash equilibria. Theoretical Computer Science
  412(22):2296 -- 2311, \doi{https://doi.org/10.1016/j.tcs.2011.01.008},
  \urlprefix\url{http://www.sciencedirect.com/science/article/pii/S030439751100034X}

\bibitem[{Bonanno(2008)}]{Bonanno2008}
Bonanno G (2008) A syntactic approach to rationality in games with ordinal
  payoffs. In: Bonanno G, van~der Hoek W, Wooldridge M (eds) Logic and the
  Foundations of Game and Decision Theory, Amsterdam University Press

\bibitem[{Capraro and Halpern(2015)}]{Capraro2015}
Capraro V, Halpern JY (2015) Translucent players: Explaining cooperative
  behavior in social dilemmas. In: Proceedings Fifteenth Conference on
  Theoretical Aspects of Rationality and Knowledge, {TARK} 2015, Carnegie
  Mellon University, Pittsburgh, USA, June 4-6, 2015., pp 114--126,
  \doi{10.4204/EPTCS.215.9},
  \urlprefix\url{https://doi.org/10.4204/EPTCS.215.9}

\bibitem[{Colbeck(2017)}]{Colbeck2017}
Colbeck R (2017) Bell inequalities. Tech. rep., University of York

\bibitem[{Colman et~al(2011)Colman, K{\"o}rner, Musy, and
  Tazda{\"\i}t}]{Colman2011}
Colman AM, K{\"o}rner TW, Musy O, Tazda{\"\i}t T (2011) Mutual support in
  games: Some properties of berge equilibria. Journal of Mathematical
  Psychology 55(2):166--175

\bibitem[{Ding et~al(2015)Ding, Wellman, Wang, Fu, and Lee}]{Ding2015}
Ding XP, Wellman HM, Wang Y, Fu G, Lee K (2015) Theory-of-mind training causes
  honest young children to lie. Psychological Science 26(11):1812--1821

\bibitem[{Dupuy(1992)}]{Dupuy1992}
Dupuy JP (1992) {Two Temporalities, Two Rationalities: A New Look At Newcomb's
  Paradox}. Economics and Cognitive Science, Elsevier pp 191--220

\bibitem[{Dupuy(2000)}]{Dupuy2000}
Dupuy JP (2000) {Philosophical Foundations of a New Concept of Equilibrium in
  the Social Sciences: Projected Equilibrium}. Philosophical Studies
  100:323--356

\bibitem[{Dupuy(2007)}]{Dupuy2007}
Dupuy JP (2007) {{Rational Choice Before The Apocalypse}}. In: Political Theory
  Workshop, Stanford University

\bibitem[{Einstein et~al(1935)Einstein, Podelsky, and Rosen}]{Einstein1935}
Einstein A, Podelsky B, Rosen N (1935) Can quantum-mechanical description of
  physical reality be considered complete? Physical Review 47(10):777--780,
  \doi{10.1103/PhysRev.47.777}

\bibitem[{Evans and Lee(2013)}]{Evans2013}
Evans AD, Lee K (2013) Emergence of lying in very young children. Developmental
  psychology 49(10):1958

\bibitem[{Foster(1983)}]{Foster1983}
Foster JE (1983) An axiomatic characterization of the theil measure of income
  inequality. Journal of Economic Theory 31(1):105--121

\bibitem[{Fourny(2018)}]{Fourny2018b}
Fourny G (2018) {{Kripke Semantics of the Perfectly Transparent Equilibrium}}.
  Tech. rep., ETH Zurich

\bibitem[{Fourny(2019{\natexlab{a}})}]{Fourny2019b}
Fourny G (2019{\natexlab{a}}) Contingent free choice: On extending quantum
  theory to a contextual, deterministic theory with improved predictive power.
  Tech. rep., ETH Z{\"u}rich

\bibitem[{Fourny(2019{\natexlab{b}})}]{Fourny2019a}
Fourny G (2019{\natexlab{b}}) Perfect prediction in minkowski spacetime:
  Perfectly transparent equilibrium for dynamic games with imperfect
  information. Tech. rep., ETH Z{\"u}rich

\bibitem[{Fourny et~al(2018)Fourny, Reiche, and Dupuy}]{Fourny2018}
Fourny G, Reiche S, Dupuy JP (2018) {{Perfect Prediction Equilibrium}}. The
  Individual and the Other in Economic Thought: An Introduction, Routledge pp
  209--257

\bibitem[{Gibbard and Harper(1978)}]{Gibbard1978}
Gibbard A, Harper W (1978) Counterfactuals and two kinds of expected utility.
  Foundations and Applications of Decision Theory

\bibitem[{Greaves and Lederman(2018)}]{Greaves2018}
Greaves H, Lederman H (2018) Extended preferences and interpersonal comparisons
  of well-being. Philosophy and Phenomenological Research 96(3):636--667,
  \doi{10.1111/phpr.12334},
  \urlprefix\url{https://onlinelibrary.wiley.com/doi/abs/10.1111/phpr.12334},
  \eprint{https://onlinelibrary.wiley.com/doi/pdf/10.1111/phpr.12334}

\bibitem[{Halpern and Pass(2018)}]{Halpern:2013aa}
Halpern JY, Pass R (2018) {Game theory with translucent players}. International
  Journal of Game Theory 47(3):949--976, \doi{10.1007/s00182-018-0626-x}

\bibitem[{Harari(2015)}]{Harari2015}
Harari YN (2015) Homo Deus: A Brief History of Tomorrow. Harvill Sekker

\bibitem[{Hofstadter(1983)}]{Hofstadter1983}
Hofstadter D (1983) {Dilemmas for Superrational Thinkers, Leading Up to a
  Luring Lottery}. Scientific American

\bibitem[{Kadar et~al(2015)Kadar, Zanni, Vogels, and Cvijikj}]{Kadar2015}
Kadar C, Zanni G, Vogels T, Cvijikj IP (2015) Towards a burglary risk profiler
  using demographic and spatial factors. In: Wang J, Cellary W, Wang D, Wang H,
  Chen SC, Li T, Zhang Y (eds) Web Information Systems Engineering -- WISE
  2015, Springer International Publishing, Cham, pp 586--600

\bibitem[{Kochen and Specker(1967)}]{Kochen1967}
Kochen S, Specker E (1967) {{The problem of hidden variables in quantum
  mechanics.}} {{Journal of Mathematics and Mechanics}} 17:59--87

\bibitem[{Kripke(1963)}]{Kripke1963}
Kripke SA (1963) Semantical considerations on modal logic. Acta Philosophica
  Fennica 16(1963):83--94

\bibitem[{Kripke(1965)}]{Kripke1965}
Kripke SA (1965) Semantical Analysis of Modal Logic II. Non-Normal Modal
  Propositional Calculi. 1, North Holland

\bibitem[{Leibniz(1710)}]{Leibniz1710}
Leibniz G (1710) {{Essais de Th{\'e}odic{\'e}e sur la bont{\'e} de Dieu, la
  libert{\'e} de l'homme et l'origine du mal}}. Fran{\c c}ois Changuio

\bibitem[{Lewis(1973)}]{Lewis1973}
Lewis D (1973) Counterfactuals. Harvard University Press

\bibitem[{Lewis(1979)}]{Lewis1979}
Lewis D (1979) Prisoner's dilemma is a newcomb problem. Philosophy and Public
  Affairs 8(3)

\bibitem[{Meyer(1999)}]{Meyer1999}
Meyer DA (1999) Quantum strategies. Physical Review Letters 82(5):1052

\bibitem[{Nash(1951)}]{JNNCG}
Nash J (1951) {Non-cooperative Games}. Annals of Mathematics 54:286 -- 295

\bibitem[{Premack and Woodruff(1978)}]{Premack1978}
Premack D, Woodruff G (1978) Does the chimpanzee have a theory of mind?
  Behavioral and brain sciences 1(4):515--526

\bibitem[{Rantala(1982)}]{Rantala1982}
Rantala V (1982) Impossible world semantics and logical omniscience. Acta
  Philosophica Fennica 35

\bibitem[{Renner and Colbeck(2011)}]{Renner2011}
Renner R, Colbeck R (2011) {{No extension of quantum theory can have improved
  predictive power.}} {{Nat Commun}} 2(411)

\bibitem[{Roddenberry(1969)}]{Roddenberry1969}
Roddenberry G (1969) {{Star Trek, The Next Generation}}

\bibitem[{Rong and Halpern(2014)}]{Rong2014}
Rong N, Halpern JY (2014) Cooperative equilibrium: {A} solution predicting
  cooperative play. CoRR abs/1412.6722,
  \urlprefix\url{http://arxiv.org/abs/1412.6722}, \eprint{1412.6722}

\bibitem[{Sen(2018)}]{Sen2018}
Sen A (2018) Collective choice and social welfare. Harvard University Press

\bibitem[{Shiffrin et~al(2009)Shiffrin, Lee, and Zhang}]{Shiffrin2009}
Shiffrin R, Lee M, Zhang S (2009) {Rational Games - `Rational' stable and
  unique solutions for multiplayer sequential games.} Tech. rep., Indiana
  University

\bibitem[{Sigmund(2011)}]{Sigmund2011}
Sigmund K (2011) Introduction to evolutionary game theory. Evolutionary Game
  Dynamics, K Sigmund, ed 69:1--26

\bibitem[{Stalnaker(1994)}]{Stalnaker1994}
Stalnaker R (1994) On the evaluation of solution concepts. Theory and Decision
  37(1):49--73, \doi{10.1007/BF01079205},
  \urlprefix\url{https://doi.org/10.1007/BF01079205}

\bibitem[{Stalnaker(1968)}]{Stalnaker1968}
Stalnaker RC (1968) A theory of conditionals. Americal Philosophical Quarterly
  pp 98--112

\bibitem[{Stalnaker(1972)}]{Stalnaker1972}
Stalnaker RC (1972) {{Letter to David Lewis}}

\bibitem[{Sulser(2019)}]{Sulser2019}
Sulser F (2019) A data-driven exploration of non-nashian game theory. Master's
  thesis 235, ETH Z{\"u}rich, Systems Group

\bibitem[{Tennenholtz(2004)}]{Tennenholtz2004}
Tennenholtz M (2004) Program equilibrium. Games and Economic Behavior 49(2):363
  -- 373, \doi{https://doi.org/10.1016/j.geb.2004.02.002},
  \urlprefix\url{http://www.sciencedirect.com/science/article/pii/S0899825604000314}

\bibitem[{Theil(1996)}]{Theil1996}
Theil H (1996) Studies in global econometrics, vol~30. Springer Science \&amp;
  Business Media

\bibitem[{Tohm{\'e} and Viglizzo(2019)}]{Thome2019}
Tohm{\'e} FA, Viglizzo ID (2019) {Superrational types}. Logic Journal of the
  IGPL \doi{10.1093/jigpal/jzz007},
  \urlprefix\url{https://doi.org/10.1093/jigpal/jzz007}, jzz007,
  \eprint{http://oup.prod.sis.lan/jigpal/advance-article-pdf/doi/10.1093/jigpal/jzz007/28284088/jzz007.pdf}

\bibitem[{Weirich(2016)}]{Weirich2016}
Weirich P (2016) Causal decision theory. Stanford Encyclopedia of Philosophy

\end{thebibliography}

\end{document}